\documentclass[twoside,11pt]{article}
\pdfoutput=1

\usepackage{blindtext}

\usepackage[preprint]{jmlr2e}
\usepackage{amsmath}
\usepackage{mathrsfs}
\usepackage[ruled,vlined]{algorithm2e}
\usepackage{amsmath}
\usepackage{mathrsfs}
\usepackage[ruled,vlined]{algorithm2e}
\usepackage{times}
\usepackage{bm}
\usepackage{amssymb}
\usepackage{verbatim}
\usepackage{multirow}
\usepackage{makecell}
\usepackage{float} 
\usepackage{graphicx,subcaption,cleveref}
\usepackage{array}

\usepackage{comment}

\DeclareMathOperator\tr{tr}
\DeclareMathOperator*\argmin{\arg\min}

\def\bA{{A}}

\def\Ex{{E}}
\def\bt{{t}}
\def\bT{{T}}
\def\pr{{\rm{pr}}}

\def\bK{{K}}

\def\bT{{T}}

\def\bW{{W}}

\def\bY{{Y}}

\def\bx{{x}}
\def\by{{y}}

\def\bGamma{{\Gamma}}
\def\bnu{{\nu}}

\def\bSigma{{\Sigma}}
\def\bsigma{{\sigma}}
\def\balpha{{\alpha}}
\def\bgamma{{\gamma}}

\def\bTheta{{\Theta}}

\def\bxi{{\xi}}

\def\bl{{l}}
\def\bmu{{\mu}}

\def\vecone{{\rm vech}}
\def\vectwo{{\rm vech_2}}

\def\R{{ R}}
\def\D{{ D}}
\def\F{{ F}}
\def\bH{{ H}}
\def\trans{^{\rm T}}

\usepackage{lastpage}
\jmlrheading{23}{2022}{1-\pageref{LastPage}}{1/21; Revised 5/22}{9/22}{21-0000}{J.TANG ET AL.}

\ShortHeadings{Network Analysis of Count Data from Mixed Populations}{J.TANG ET AL.}
\firstpageno{1}

\begin{document}

\title{Network Analysis of Count Data from Mixed Populations}

\author{\name Junjie Tang \email junjie.tang@pku.edu.cn \\
	    \name Changhu Wang \email wangch156@pku.edu.cn \\
	   	\name Feiyi Xiao \email xiaofeiyi1217@pku.edu.cn \\
	   	\name Ruibin Xi \email ruibinxi@math.pku.edu.cn \\
       \addr School of Mathematical Sciences\\
       Peking University\\
       Beijing, 100871, CHN}

\editor{My editor}

\maketitle

\begin{abstract}
In applications such as gene regulatory network analysis based on single-cell RNA sequencing data, samples often come from a mixture of different populations and each population has its own unique network. Available graphical models often assume that all samples are from the same population and share the same network. One has to first cluster the samples and use available methods to infer the network for every cluster separately. However, this two-step procedure ignores uncertainty in the clustering step and thus could lead to inaccurate network estimation. Motivated by these applications, we consider the mixture Poisson log-normal model for network inference of count data from mixed populations. The latent precision matrices of the mixture model correspond to the networks of different populations and can be jointly estimated by maximizing the lasso-penalized log-likelihood. Under rather mild conditions, we show that the mixture Poisson log-normal model is identifiable and has the positive definite Fisher information matrix. Consistency of the maximum lasso-penalized log-likelihood estimator is also established. To avoid the intractable optimization of the log-likelihood, we develop an algorithm called VMPLN based on the variational inference method. Comprehensive simulation and real single-cell RNA sequencing data analyses demonstrate the superior performance of VMPLN. 
\end{abstract}


\section{Introduction}
\label{sec:intro}
Graphical models \citep{drton2017structure} such as the Gaussian graphical model \citep{meinshausen2006high,friedman2008sparse} have been widely applied to many different fields for identifying key interactions between random variables \citep{farasat2015probabilistic,wille2004sparse,dobra2011copula}. These graphical models usually assume that all samples, or at least the samples under a known condition, are sampled from the same population and thus have the same network. However, in applications such as the recent single-cell RNA sequencing (scRNA-seq) studies \citep{aibar2017scenic,specht2017leap,chan2017gene,aibar2017scenic,song2022single}, samples often come from a mixture of different populations and each population has its own unique network. In scRNA-seq studies, samples are single cells of different cell types and each cell type has its unique gene regulatory network. Gene expressions of single cells are measured as counts of short reads or unique molecular identifiers. Because the cell identities are unknown, one has to first assign single cells to different cell types (e.g. by clustering) and then estimates the gene regulatory networks using available methods. This two-step procedure can provide accurate network estimation if different populations are well separated. If, instead, different populations have a higher mixing degree, a large proportion of samples cannot be confidently assigned with a label and the incorrect label assignment could seriously influence the performance of network inference (See Figure \ref{fig:total} for an example).

\begin{figure}[htbp]
	\centering
	\includegraphics[width = 1.00\textwidth, height=0.35\textwidth]{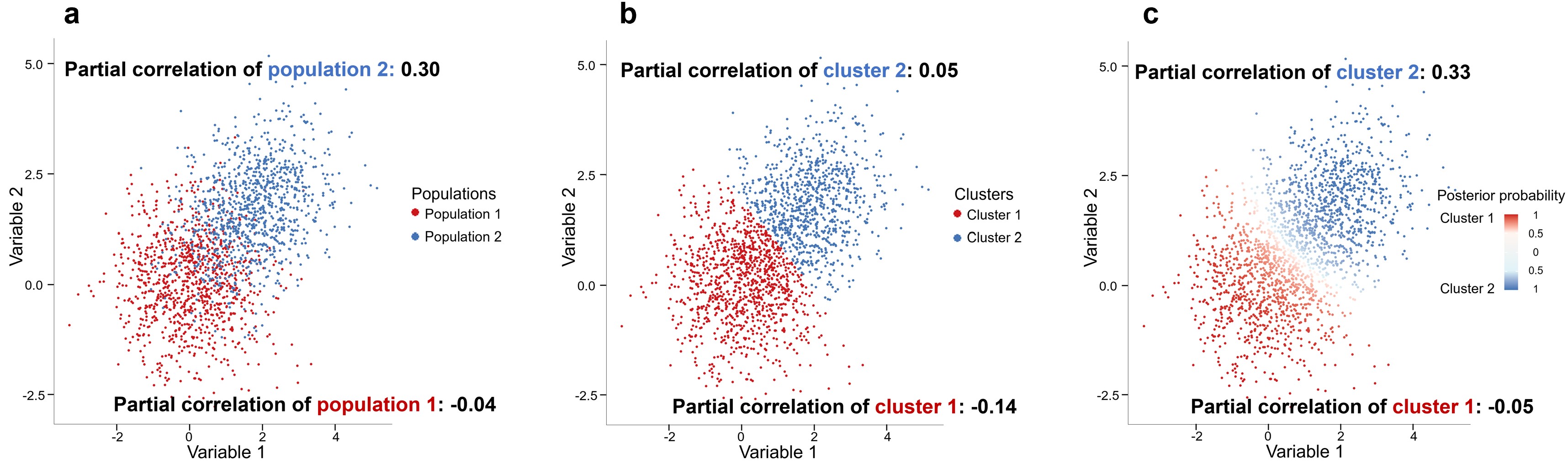}
	\caption{An example of data from a two-dimensional mixture Gaussian distribution. We sample 2000 observations from the Gaussian mixture model $0.5\ \mbox{N}\left(\mu = (0,0) ,\Sigma =  \Theta_1^{-1}\right) + 0.5\ \mbox{N}\left( \mu = (1.7,1.7) , \Sigma = \Theta_2^{-1}\right)$. The diagonal elements of precision matrices $\Theta_1$ and $\Theta_2$ are 1, and the non-diagonal elements of $\Theta_1$ and $\Theta_2$ are 0 and $-0.3$, respectively. (a) The scatter plot of the 2000 samples with true labels (red: population 1, blue: population 2). Using the true label, as expected, samples in population 1 demonstrate a partial correlation very close to 0, and samples in population 2 have a large nonzero partial correlation (0.30). (b) The samples are first clustered using the K-means algorithm. Partial correlations are calculated for each cluster and are far away from the true partial correlations, which would lead to incorrect network estimates. (c) Parameters are estimated by maximizing the likelihood of the mixture Gaussian distribution. The estimates of the partial correlations are now close to the true values.}
	\label{fig:total}
\end{figure}

Motivated by the gene regulatory network inference problem in scRNA-seq studies, we consider developing a network inference method for count data from mixed populations using mixture models. One major advantage of mixture models for network inference is that sample identities need not be predetermined before network inference. Instead, mixture models can allow joint analyses of clustering and network inference, and thus could give better network estimation when different populations are poorly separated (See Figure \ref{fig:total} for an example). Available graphical models for count data include Poisson graphical models and their extensions \citep{yang2012graphical,allen2013local}, negative binomial graphical models \citep{park2021negative} and Poisson log-normal models \citep{Wu2018,chiquet2019variational,silva2019multivariate}. The Poisson and negative binomial graphical models can only allow negative interactions, and the extension of the Poisson graphical models have no explicit form of the joint distribution and cannot adequately account for over-dispersion in the data, such as in the scRNA-seq data \citep{ziegenhain2017comparative}. We therefore consider the mixture of Poisson log-normal models for network inference of count data from mixed populations.

A non-negative integer random vector $Y = (Y_1,\dots,Y_p)^T \in \mathbb{R}^p$ follows a Poisson log-normal distribution, if conditional on a latent normal random vector $X = (X_1,\dots,X_p)^T \sim \mbox{N}(\mu,\Sigma)$, each element $Y_j$ of $Y$ independently follows the univariate Poisson distribution $\mbox{Poisson}\left\{\exp(X_j)\right\}$ ($j = 1,\dots,p$). Similar to the Gaussian graphical model, the network of the Poisson log-normal model is the precision matrix $\Theta = {\Sigma}^{-1}$ of the latent variable $X$. The mixture Poisson log-normal model is a mixture of $G$ different Poisson log-normal models, and the underlying $G$ latent precision matrices are the networks of different populations. Assuming that the networks are sparse, we can maximize the lasso-penalized log-likelihood of the mixture Poisson log-normal model to estimate the networks. The mixture Poisson log-normal model has been used for clustering count data \citep{silva2019multivariate} and EM algorithms coupled with MCMC steps are developed to maximize its computational intractable log-likelihood. However, theoretical properties of the mixture Poisson log-normal model are not studied, and EM algorithms with MCMC steps are computationally very expensive.

Here, we propose to use these mixture models for network inference. We establish the basic properties of the mixture Poisson log-normal model including its identifiability and the positive definiteness of its Fisher information matrix. We further show that the network estimator by maximizing the lasso-penalized log-likelihood is consistent. All technical proofs are in appendix Section \ref{Appendix:Technicalproofs}. We adopt the variational inference approach \citep{jordan1999introduction,wainwright2008graphical} and develop a more efficient algorithm called VMPLN for the network inference. We compare VMPLN with popular graphical methods and state-of-the-art gene regulatory network inference methods using simulation and real data analyses. We also demonstrate an application of VMPLN to a large scRNA-seq dataset from patients infected with severe acute respiratory syndrome coronavirus 2 (SARS-CoV-2). All data used and corresponding source codes are available and can be accessed at \url{https://doi.org/10.5281/zenodo.7069698}.

\section{Model and Theoretical Properties}
\label{sec:Model}
In this section, we describe the mixture Poisson log-normal model for network inference of count data from mixed populations, and then establish the theoretical properties of the mixture Poisson log-normal model as well as the consistency of network estimator by maximizing the lasso-penalized log-likelihood.
\subsection{The Mixture Poisson Log-normal Model of Count data from Mixed Populations}
Let $Y_{i}=\left(Y_{i 1}, \dots, Y_{i p}\right)^T$ be the $i$th observation ($i=1,\dots,n$), where $Y_{ij}$'s are all non-negative integers. The $n$ samples belong to $G$ different populations. Conditional on latent variables $X_{ij}$, we assume that $Y_{ij}$'s are independent Poisson variables with  parameters $\lambda_{i j} = l_i \exp\left(X_{ij}\right)$, where $l_i$'s are known scaling factors. The underlying population $Z_i$ of the $i$th sample follows a multinomial distribution $\mbox{Multinomial} (1,\pi)$, where $\pi = (\pi_1,\dots,\pi_{G})^T$ is the proportion parameter representing the composition of the mixture components. Given $Z_i = g$ ($g=1,\dots,G$), the latent vector $X_i = \left(X_{i 1}, \dots, X_{i p}\right)^T$ is normally distributed with a mean $\mu_g\in \mathbb{R}^p$ and a covariance matrix $\Theta_g^{-1}$. In summary, the mixture Poisson log-normal model can be written as, 

\begin{equation}\label{equ0}
	\begin{aligned}
		Y_i \mid X_i &\sim \prod_{j=1}^p\mbox{Poisson}\left\{l_i\ \exp \left(X_{i j}\right)\right\}, \\
		X_i \mid Z_i = g &\sim \mbox{N}\left( \mu_g ,{\Theta_g}^{-1} \right), \Theta_g \succ 0,\\
		Z_i &\sim \mbox{Multinomial} (1,\pi),
	\end{aligned}
\end{equation}
where $\Theta_g \succ 0$ means that $\Theta_g$ is positive definite. In scRNA-seq data, $Y_{ij}$ is the observed expression of the $j$th gene in the $i$th cell, $\exp\left(X_{ij}\right)$ represents the underlying ``true'' expression, and $l_i$ is the library size of the $i$th cell and can be readily estimated \citep{hafemeister2019normalization,lun2016pooling}.  

Denote $\theta = \big(\pi, \mu=\{\mu_g\}_{g = 1}^G,\Theta = \{\Theta_g \}_{g = 1}^G\big)$ as unknown parameters of the model (\ref{equ0}). Let $p(Y_i \mid X_i) = \prod_{j = 1}^{p} \left[\left\{l_i \exp \left(X_{i j}\right)\right\}^{Y_{i j}} \exp\{-l_i \exp \left(X_{i j}\right)\}(Y_{i j }!)^{-1}\right]$ be the conditional probability mass function of $Y_i$ given $X_i$. The conditional density function $p(X_i \mid Z_i;\mu,\Theta)$ of $X_i$ given $Z_i$ can be written as $p(X_i \mid Z_i;\mu,\Theta) = \prod_{g = 1}^{G} \left\{p(X_i;\mu_g,\Theta_g)\right\}^{I\left(Z_i = g\right)}$, where $p(\cdot\ ;\mu_g,\Theta_g)$ is the density of $\mbox{N}\left( \mu_g ,{\Theta_g}^{-1} \right)$. Denote $p(Z_i ; \pi) = \prod_{g = 1}^{G} \pi_{g}^{I(Z_i = g)}$ as the probability mass function of the multinomial distribution. The log-likelihood of the mixture model (\ref{equ0}) is
\begin{equation}\label{equ_likelihood}
	\ell_n(\theta) = \sum_{i = 1}^{n} \log\left\{p(Y_i;\theta)\right\} = \sum_{i = 1}^{n}\log \iint p(Y_i \mid X_i) p(X_i \mid Z_i;\mu,\Theta) p(Z_i;\pi) \mathrm{d} X_i \mathrm{d} Z_i\ ,
\end{equation}
where $p(Y_i;\theta)$ is the marginal probability mass function of $Y_i$. The precision matrices $\Theta_g$'s represent the population-specific network. Let $\Theta_{g,l m}$ be the $(l,m)$th element of $\Theta_g$. The networks are sparse and can be estimated by minimizing the lasso-penalized negative log-likelihood 
\begin{equation}\label{penalizedloglike}
	- n^{-1}\ell_n(\theta)+\lambda_{n} \sum_{g=1}^{G}\| {\Theta_{g}}\|_{1, \text { off }},
\end{equation}
where $\lambda_{n}>0$ is a tuning parameter and $\| {\Theta_{g}}\|_{1, \text { off }} = \sum_{l \neq m} |{\Theta_{g,l m}}|$  is the off-diagonal $l_1$-norm of $\Theta_{g}$. In the following, we first establish the consistency of the estimator obtained by minimizing (\ref{penalizedloglike}) and then derive an algorithm for estimating the precision matrices.

\subsection{Theoretical Properties}
In this section, we always assume that the true means $\mu_g^*$ and proportions $\pi_g^*$ ($g=1,\dots,G$) are known. Let $\nu_g = {\rm vech}(\Theta_g) $ be the vectorization of the precision matrix $\Theta_g$ (appendix Section \ref{Appendix:Technicalproofs_notation}) and ${\nu} =  (\nu_1^{\rm T}, \ldots, \nu_G^{\rm T})^{\rm T}$ . In this case, the log-likelihood $\ell_n(\theta)$ can be viewed as a function of ${\nu}$, also denoting as $\ell_n(\nu)$, and we consider the estimator $\hat{\nu}_n$ that minimizes $-n^{-1}\ell_n(\nu) + \lambda_{n} \sum_{g=1}^{G}\| {\Theta_{g}}\|_{1, \text { off }}$ subject to $\Theta_{g} \succ 0\  (g=1,\dots,G)$.  Denote $\nu^*$ as the true value of the unknown parameter $\nu$, $S(\nu) = \{i  \mid  \ \nu_{i} \neq 0\}$ as the support of $\nu$, and $S^{*} = S({\nu}^{*})$. Suppose that $Y$ follows the mixture model (\ref{equ0}) with its log-likelihood function $\ell(\nu,Y)$. Denote $\Gamma^* = \Gamma(\nu^*) = -{\rm E}\left\{\partial^2\ell(\nu, Y)/\partial \nu \partial \nu^T\right\}\mid_{\nu = \nu^*}$ as the Fisher information matrix at $\nu^*$, and ${\Gamma^*}_{T_1 T_2}$ as the submatrix of ${\Gamma^*}$ with rows and columns index by sets $T_1$ and $T_2$, respectively. Before presenting the theoretical properties, we give the following conditions:\\
\textbf{Condition C1}. The eigenvalues of the precision matrices are bounded in $[m,M]$, where $0<m<M<\infty$.\\
\textbf{Condition C2}. The scaling factors $l_i>0$ ($i = 1,\dots,n$) are independent and identically distributed random variables with a bounded support.\\
\textbf{Condition C3}. The true mean vectors $\mu_g^*$ ($g = 1,\dots,G$) are bounded and different from each other.\\
\textbf{Condition C4} (The irrepresentability condition).  $||\Gamma^*_{{S^{*}}^c S^{*}}{(\Gamma^{*}_{S^{*} S^{*}})}^{-1}||_{\infty} < 1$. 

Denote $\lambda_{min}(\Theta_g)$ and $\lambda_{max}(\Theta_g)$ as the minimum and maximum eigenvalue of $\Theta_g$. Define $\mathcal{D} = \{ {\nu} =  (\nu_1^{\rm T}, \nu_2^{\rm T}, \ldots, \nu_G^{\rm T})^{\rm T}|~\nu_g = {\rm vech}(\Theta_g), m \leq \lambda_{min}(\Theta_g)\leq\lambda_{max}(\Theta_g)\leq M \}$ and $ \kappa =  \lambda_{min}(\Gamma^*)$ be the minimum eigenvalue of the Fisher information matrix at $\nu^{*}$. Condition C1-C2 are commonly used in the literature \citep{cai2011constrained, li2020transfer}. Condition C3 is to ensure that different components of the mixture model (\ref{equ0}) can be distinguished from each other. Under Condition C3, the mixture model (\ref{equ0}) is identifiable and its Fisher information matrix $\Gamma^*$ is positive definite, and thus $\kappa>0$. The irrepresentability condition C4 is also commonly used \citep{zhao2006model, ravikumar2011high}. Based on the above conditions, we present the theoretical properties of the mixture model (\ref{equ0}) and the estimator $\hat{\nu}_n$ in the following theorems.

\begin{theorem} \label{thm:basic}
	Under Condition C1-C3, the mixture Poisson log-normal model  (\ref{equ0}) is identifiable, and its Fisher information matrix  $\Gamma^*$ at $\nu^*$ is positive definite.
\end{theorem}

Theorem \ref{thm:basic} establishes basic properties of the mixture model (\ref{equ0}) and ensures that it is well-behaved under rather mild conditions. The proof of this theorem is nontrivial because the Poisson log-normal distribution has no finite moment-generating function and its density function is rather complex. However, its moments are finite and have closed forms. We use its moments to prove Theorem \ref{thm:basic}. For the identifiability, the basic idea of the proof is that  identifiability of a mixture model is equivalent to linear independence of its components. Using moments of the Poisson log-normal distribution, we can show that only the zero vector can make the linear combination of the components of the mixture model (\ref{equ0}) as zero. To prove the positive definiteness of the Fisher information matrix, we also use the moments of the Poisson log-normal distribution and convert the problem to showing that a set of equations only have zero solutions. Based on this result, we can further prove the consistency and the sign consistency of the estimator $\hat{\nu}_n$.
\begin{theorem} \label{thmConvergenceRate}
	Under Condition C1-C3, we have
	$$pr\left[||\hat{\nu}_n-\nu^*||_2 \leq  (3/\kappa) \left\{G{p(p+1)}/2\right\}^{1/2}\big(n^{-1} ||\nabla{\ell}_n(\nu^*)||_{\infty} + 2{\lambda_n}\big)\right] \rightarrow 1, \mbox{ as } n \rightarrow \infty,$$
	where $\nabla{\ell}_n(\nu^*)$ is the gradient of $\ell_n(\nu)$ at $\nu^*$ and $\lambda_n > 0$ is the regularization parameter.
\end{theorem}

\begin{theorem} \label{signconsistency}
	Under Condition C1-C4, choosing $\lambda_n > 0$ such that $\lambda_n \rightarrow 0$ and ${n^{1/2}\lambda_n} \rightarrow \infty$, we have $pr\left\{S(\hat{\nu}_n) = S^{*}\right\} \rightarrow 1 \mbox{ as } n \rightarrow \infty.$
\end{theorem}
Theorem \ref{signconsistency} says that if we choose $\lambda_n > 0$ such that it does not converge to 0 too fast, e.g. $\lambda_n = \left(\log n / n\right)^{1/2}$, $\hat{\nu}_n$ can consistently recover the nonzero elements of $\nu^*$. Theorem \ref{thmConvergenceRate} and \ref{signconsistency} imply that, for the Poisson log-normal model, the network estimated by minimizing its lasso-penalized negative log-likelihood is consistent. All proofs are given in appendix Section \ref{Appendix:Technicalproofs}.

\section{Variational Inference for the Mixture Poisson Log-normal Model}
\label{sec:Algorithm}
The log-likelihood (\ref{equ_likelihood}) of the mixture Poisson log-normal model involves an intractable integration and thus directly minimizing (\ref{penalizedloglike}) is computationally very difficult. We therefore adopt the variational inference approach to estimate the networks \citep{jordan1999introduction, wainwright2008graphical}. We approximate the log-likelihood by the evidence low bound $\ell_{\text{E}}\left(\eta,\theta \right)$ and estimate $\theta$ by minimizing $-\ell_{\text{E}}\left(\theta,\eta \right)+\lambda_{n} \sum_{g=1}^{G}\left\| {\Theta_{g}}\right\|_{{1, \text { off }}}$, where $\eta\in \mathscr{H}$ is the parameter of the variational distribution family $\mathscr{L}=\{q(X,Z;\eta): \eta \in \mathscr{H}\}$. For $\eta\in \mathscr{H}$,  the evidence low bound $\ell_{\text{E}}\left(\eta,\theta \right)$ is defined as $\ell_{\text{E}}\left(\eta,\theta \right) = {\mbox{E}}_{q(X,Z;\eta)}\left\{\log p(X,Z, Y;\theta) - \log q(X,Z;\eta)\right\}$. 

For computational considerations, we consider the following variational distribution family. Conditional on $Z_i=g$, this distribution family assumes that $X_{ij}$ ($j=1,\dots,p$) are independent normal variables with a mean $M_{g,ij}$ and a variance $S_{g,ij}$. The distribution of $Z_i$ is a multinomial distribution with proportion parameters $P_i = (P_{i1},\dots,P_{iG})^T$. Denote $M_{g}=(M_{g,ij})_{n\times p}$, $S_{g}=(S_{g,ij})_{n\times p}$ and $P = (P_1,\dots,P_n)^T$. The variational parameters are $\eta = \left\{\eta_g\right\}_{g=1}^{G} = \left\{M_{g}, S_{g}, \{P_{ig}\}_{i = 1}^{n}\right\}_{g = 1}^{G}$ with $\eta\in\mathscr{H} = \left\{\eta \big| \ S_{g,ij}>0,~P_{ig}\geq0,~\sum_{g=1}^G P_{ig} =1\right\}$. Thus, the variational distribution family is
$$
\mathscr{Q} = \left\{q(X,Z;\eta) = \prod_{i=1}^{n} \left\{q(Z_i;P_i) \prod_{j=1}^{p} \prod_{g=1}^{G} q\big(X_{i j} \mid Z_{i} = g; M_{g,i j},S_{g,i j}\big)\right\},\ \eta \in \mathscr{H}\right\},
$$
where $q\big(X_{i j} \mid Z_{i} = g; M_{g,i j},S_{g,i j}\big)$ is the density of $N(M_{g,i j},S_{g,i j})$ and $q(Z_i;P_i)$ is the density of $\mbox{Multinomial} (1,P_i)$. 

Given two matrices $A$ and $B$, let  $A \odot B$ be their Hadamard product. Denote $A_{i \cdot}$ and $A_{\cdot j}$ as the $i$th row and the $j$th column vectors of $A$, respectively. Given a vector $c$, define $D(c)$ as the diagonal matrix whose diagonal elements are $c$. Denote $l = \left(l_1,\dots,l_n\right)^T$, $\Sigma_{g i} = (M_{g,i\cdot} - \mu_g){(M_{g, i\cdot} - \mu_g)}^T + D \left(S_{g ,i\cdot}\right)$, $F_1\left(l_i,M_{g,i j},S_{g,i j}\right) = \exp\left(M_{g,i j} + 2^{-1}S_{g,i j} + \log l_i\right)$, $F_2\left(\Theta_{g},M_{g,i \cdot},S_{g,i \cdot},\mu_{g}\right) = $\\
$2^{-1} \left\{\log \det \Theta_g - \tr (\Theta_g \Sigma_{g i})\right\}$ and $\theta_g = \left(\pi_{g}, \mu_g, \Theta_g\right)$ as the unknown model parameters of the $g$th population. With the variational distribution family $\mathscr{L}$, the evidence low bound can be written as $\ell_{\text{E}}\left(\eta,\theta \right) = \sum_{g = 1}^{G} \ell_{\text{E}}^{(g)} \left( \eta_g,\theta_g \right)$ with 
$$
\ell_{\text{E}}^{(g)} \left(\eta_g,\theta_g \right) = P_{\cdot g}^{T} \left(\Lambda_{g}^{(1)} - \Lambda_{g}^{(2)} + \Lambda_{g}^{(3)}\right) 1_{p} + P_{\cdot g}^{T} \left(\Lambda_{g}^{(4)} + \Lambda_{g}^{(5)}\right) + K_g\left(Y\right),
$$
where $\Lambda_{g}^{(1)} = Y \odot M_{g}$, $\Lambda_{g}^{(2)} = \left(\Lambda_{g,i j}^{(2)}\right)_{n \times p} = \big(F_1\left(l_i,M_{g,i j},S_{g,i j}\right)\big)_{n \times p}$, $\Lambda_{g}^{(3)} = 2^{-1} \log S_{g}$, $\Lambda_{g}^{(4)} = \log (\pi_{g}) 1_{n} - \log \left(P_{\cdot g}\right)$, $\Lambda_{g}^{(5)} = \left(\Lambda_{g,i}^{(5)}\right)_{n} = \big(F_2\left(\Theta_{g},M_{g,i \cdot},S_{g,i \cdot},\mu_{g}\right)\big)_{n}$ and $K_g(Y) = \sum_{i, j} P_{i g} \left\{- \log \left(Y_{i j} !\right) + Y_{i j}\log l_i\right\}$.

In real applications, we may have prior knowledge that some node pairs cannot have direct interactions. In this case, we can directly set the corresponding edges as zero. Denote $E_p$ as the set of the edges that are priorly known to be zero. Generally, we consider the following optimization problem
\begin{equation}\label{con:VMPLNloss}
	\begin{aligned}
		\underset{\eta,\theta}{\min}\left\{-\ell_{\text{E}}\left(\theta,\eta \right)+\lambda_{n} \sum_{g=1}^{G}\left\| {\Theta_{g}}\right\|_{{1, \text { off }}}: ~ {\Theta}_{g} \succ 0,\ {\Theta_{g,l m}} = 0 ~ \mbox{for}~(l,m) \in E_{p},\ \eta \in \mathscr{H}\right\}.
	\end{aligned}
\end{equation}

\begin{algorithm}
	\footnotesize
	\caption{Framework of VMPLN.}
	\label{alg:algorithm1}
	\KwIn{Count data $Y$, the pre-estimated scaling factor $l$, the number of populations $G$, the tuning parameter $\lambda_{n}$, the optional prior set of zero edges $E_{p}$, the maximum iteration number $K>0$ and the convergence thresholds $\epsilon_{L},\epsilon_{s}$.}
	\KwOut{$\hat{\pi}, \hat{\mu}, \{\hat{\Theta}_g \}_{g = 1}^G, \{\hat{M}_{g}\}_{g = 1}^{G}, \{\hat{S}_{g}\}_{g = 1}^{G}, \hat{P}$.}
	\BlankLine
	\underline{Initialization step}. Let $\delta_{L} = 10^{6}, \delta_{s} = p(p+1)/2, k = 0$ and initialize $\theta, \eta$ as $\theta^{(0)}, \eta^{(0)}$.
	
	\While{\textnormal{ $\left(\delta_{L} > \epsilon_{L} \ \text{ or } \ \delta_{s} > \epsilon_{s}\right) \ \text{ and } \ k \leq K $ }}{
		\underline{$P$-step}. For each $\left(i,g\right) \in \mathscr{N} \times \mathscr{G}$, compute $U_{i g}^{(k)} = F_2\left(\hat{\Theta}_{g}^{(k)}, \hat{M}_{g,i \cdot}^{(k)}, \hat{S}_{g,i \cdot}^{(k)},\hat{\mu}_{g}^{(k)}\right)$ and update
		$${\hat{P}_{i g}}^{(k+1)}  = \hat{\pi}_{g}^{(k)} \exp \left({U}_{i g}^{(k)}\right) \bigg/ \sum_{l = 1}^{G} \left\{\hat{\pi}_{l}^{(k)} \exp \left({U}_{i l}^{(k)}\right)\right\}.$$
		
		\underline{$\pi$-step}. For each $g \in \mathscr{G}$, update
		$${\hat{\pi}_{g}}^{(k+1)} = n^{-1}\sum_{i = 1}^{n} {\hat{P}_{i g}}^{(k+1)}.$$
		
		\underline{$M$-step}. For each $\left(i,g\right) \in \mathscr{N} \times \mathscr{G}$, update
		\begin{equation}\label{opti_M}
			{\hat{M}_{g,i\cdot}}^{(k+1)} = \underset{M_{g,i\cdot}}{\argmin} \left\{ L_{1} \left(M_{g, i \cdot}, \hat{\Theta}_{g}^{(k)}, \hat{\mu}_g^{(k)}\right) + \sum_{j = 1}^{p} L_{2}\left(M_{g,i j}, l_i, \hat{S}_{g,i j}^{(k)}\right) \right\}.
		\end{equation}
		
		\underline{$S$-step}. For each $\left(i,g,j\right) \in \mathscr{N} \times \mathscr{G} \times \mathscr{P}$, update
		$${\hat{S}_{g,i j}}^{(k+1)} = \underset{S_{g,i j} > 0}{\argmin} \left\{F_1\left(l_i,\hat{M}_{g,i j}^{(k+1)},S_{g,i j}\right) + 2^{-1} \hat{\Theta}_{g,j j}^{(k)} S_{g,i j} - 2^{-1} \log S_{g,i j} \right\}.$$
		
		\underline{$\mu$-step}. For each $\left(g,j\right) \in \mathscr{G} \times \mathscr{P}$, update
		$$\hat{\mu}_{g j}^{(k+1)} = \sum_{i = 1}^{n} \left(\hat{P}_{i g}^{(k+1)} \hat{M}_{g,i j}^{(k+1)}\right)\bigg/ \sum_{i = 1}^{n} \hat{P}_{i g}^{(k+1)}.$$
		
		\underline{$\Theta_{g}$-step}. \ For each $g \in \mathscr{G}$, compute $\hat{\Sigma}_g^{(k+1)} = \Sigma\left(\hat{P}_{.g}^{(k+1)}, \hat{\mu}_{g}^{(k+1)},\hat{M}_{g}^{(k+1)},\hat{S}_{g}^{(k+1)} \right)$ and update
		\begin{equation*}
			\begin{aligned}
				\hat{\Theta}_{g}^{(k+1)} = \underset{\Theta_{g}}{\argmin} \bigg\{&- 2^{-1} \log \det \Theta_{g } + 2^{-1} \tr\big(\Theta_{g} \hat{\Sigma}_{g}^{(k+1)}\big) + \left(\lambda_{n} / 1_{n}^{T} \hat{P}_{\cdot g}^{(k+1)} \right) \left\| {\Theta_{g}}\right\|_{{1, \text { off }}}:\\
				&{\Theta}_{g} \succ 0,\ {\Theta_{g,l m}} = 0 ~ \mbox{for}~(l,m) \in E_{p}\bigg\}.
			\end{aligned}
		\end{equation*}
		
		\underline{Evaluation step}. \ Update \begin{equation*}
			\begin{aligned}
				\delta_{s} &= \underset{g \in \mathscr{G}}{\max} \left\{\sum_{l<m} \left|\text{sign}\left(\hat{\Theta}_{g,lm}^{(k+1)}\right) - \text{sign}\left(\hat{\Theta}_{g,lm}^{(k)}\right) \right|  {\tbinom{p}{2}}^{-1} \right\},\\
				\delta_{L} &= \delta\left(\ell_{\text{E}}\left(\hat{\eta}^{(k+1)},\hat{\theta}^{(k+1)} \right),  \ell_{\text{E}}\left(\hat{\eta}^{(k)},\hat{\theta}^{(k)} \right)\right),\ \ \delta \left(a,b\right) = \left| a - b\right|/b, \ k = k+1.
			\end{aligned}
		\end{equation*}
		\BlankLine
	}
\end{algorithm}

\subsection{The Optimization Process}
\label{subsec:optiprocess}
We develop a block-wise descent algorithm called VMPLN to optimize (\ref{con:VMPLNloss}). Given initial values, we iteratively update $P$, $\pi$, $\left\{M_{g}\right\}_{g = 1}^{G}$, $\left\{S_{g}\right\}_{g = 1}^{G}$, $\mu$ and $\left\{\Theta_{g}\right\}_{g = 1}^{G}$ and terminate the iteration if changes between two successive update steps are small. The VMPLN algorithm is summarized in Algorithm \ref{alg:algorithm1}, where we have used the following notations, $\mathscr{N} = \left\{1,\dots,n\right\}, \mathscr{G} = \left\{1,\dots,G\right\}, \mathscr{P} = \left\{1,\dots,p\right\}$, $L_{1}\left(M_{g, i \cdot}, \Theta_{g}, \mu_g\right) = 2^{-1} {\left(M_{g, i \cdot} - \mu_g\right)}^T \Theta_{g} \left(M_{g, i \cdot} - \mu_g\right)$, $L_{2}\left(M_{g,i j}, l_i, S_{g,i j}\right) = - \Lambda_{g,i j}^{(1)} + F_{1}\left(l_i, M_{g,i j}, S_{g,i j}^{(k+1)}\right)$, and 
$$\Sigma \left(P_{\cdot g}, \mu_{g},M_{g},S_{g}\right) = \left[\sum_{i = 1}^{n} P_{i g} \left\{(M_{g,i\cdot} - \mu_g){(M_{g, i\cdot} - \mu_g)}^T + D \left(S_{g ,i\cdot}\right)\right\}\right] \bigg/ \sum_{i = 1}^{n} P_{i g}. $$
The parameters $P$, $\pi$ and $\mu$ all have explicit updating formulas and can be efficiently calculated. For the parameter $S$, given all other parameters, the loss function (\ref{con:VMPLNloss}) can be decomposed into a sum of $npG$ functions, each of which only involves one $S_{g,ij}$ and can be efficiently solved by the Newton-Raphson algorithm. For the network parameters $\Theta_g$ ($g=1,\dots,G$), given all other parameters, the corresponding sub-optimization problem is equivalent to solving $G$ independent Glasso problems \citep{meinshausen2006high,friedman2008sparse}. The step for updating $M_g$ ($g=1,\dots,G$) is presented in the next Subsection \ref{subsec:Mstep}.

\subsection{The Optimization of the $M$-step in Algorithm \ref{alg:algorithm1}}
\label{subsec:Mstep}
The sub-optimization problem corresponding to $M_g$ ($g=1,\dots,G$) is
$$\underset{\{M_g,\ g\in \mathscr{G}\}}{\argmin} \sum_{g \in \mathscr{G},\ i \in \mathscr{N}} P_{i g}\left\{L_{1} \left(M_{g, i \cdot}, \hat{\Theta}_{g}^{(k)}, \hat{\mu}_g^{(k)}\right) + \sum_{j = 1}^{p} L_{2}\left(M_{g,i j}, l_i, \hat{S}_{g,i j}^{(k)}\right)\right\},$$ 
which is equivalent to $nG$ independent optimization problems (\ref{opti_M}) (See Algorithm \ref{alg:algorithm1}). We develop an efficient algorithm based on the alternating direction method of multipliers algorithm \citep{boyd2011distributed} to optimize (\ref{opti_M}). Specifically, we introduce an auxiliary matrix $N_{g}$ for $M_{g}$, and denote $L_M\left(M_{g, i \cdot},N_{g, i \cdot}\right) = L_{1} \left(N_{g, i \cdot}, \hat{\Theta}_{g}^{(k)}, \hat{\mu}_g^{(k)}\right)$ $ + \sum_{j = 1}^{p} L_{2}\left(M_{g,i j}, l_i, \hat{S}_{g,i j}^{(k)}\right)$. Solving (\ref{opti_M}) is equivalent to solving the following problem

\begin{equation}\label{opti_MN}
	\underset{M_{g, i \cdot}=N_{g,i\cdot}}{\argmin} L_M \left(M_{g, i \cdot},N_{g, i \cdot}\right).
\end{equation}
The augmented Lagrangian of optimization problem (\ref{opti_MN}) is
$$
L_M \left(M_{g, i \cdot},N_{g, i \cdot}\right) +\sum_{j = 1}^{p} {\alpha}_j \left(M_{g,i j} - N_{g,i j}\right) + \rho/2 \sum_{j = 1}^{p} {\left(M_{g,i j} - N_{g,i j}\right)}^2,
$$
where $\alpha = \left(\alpha_{1},\dots,\alpha_{p}\right)$ is the Lagrangian multiplier, and $\rho$ is the step size. Given initial values, we iteratively update $M_{g, i \cdot}$, $N_{g, i \cdot}$ and $\alpha$ (Algorithm \ref{alg:algorithm2}). The optimization problem for $M_{g,i\cdot}$ can be decomposed into $p$ independent one-dimensional optimization problems, which can be easily optimized by the one-dimensional Newton-Raphson algorithm (Step 1 in Algorithm \ref{alg:algorithm2}). The optimization problem for $N_{g,i\cdot}$  has an explicit solution and involves inverting the matrix $\rho I + \hat{\Theta}_{g}^{(k)}$. We only need to invert the matrix $\rho I + \hat{\Theta}_{g}^{(k)}$ once. In total, the $M$-step of Algorithm \ref{alg:algorithm1} involves $nG$ optimization problems (\ref{opti_M}), but we only need to calculate $G$ $p\times p$ matrix inversions in the $nG$ runs of Algorithm \ref{alg:algorithm2}.

\begin{algorithm}
	\footnotesize
	\caption{ADMM algorithm for updating $M_{g,i \cdot}$.}
	\label{alg:algorithm2}
	\KwIn{Count data $Y$, the pre-estimated scaling factor $l$, the current estimation of parameters $\hat{S}_{g,i \cdot}, \hat{\Theta}_{g}, \hat{\mu}_g$, the maximum iteration number $T>0$ and the convergence threshold $\epsilon_{M}$.}
	\KwOut{$\hat{M}_{g,i \cdot}$.}
	\underline{Initialization step}. Let $t = 0, \delta_{M} = 10^{6}$, and initialize $\hat{M}_{g,i\cdot}^{\left(0\right)} = \hat{N}_{g,i\cdot}^{\left(0\right)} = \hat{M}_{g,i\cdot}$, $\hat{\alpha}^{(0)} = 0$.
	
	\While{$t \leq T$ \text{ and } $\delta_{M} > \epsilon_{M}$}{
		
		\underline{Step 1}. For each $j \in \mathscr{P}$, update $$\hat{M}_{g,i j}^{(t+1)} = \underset{M_{g,i j}}{\argmin} \bigg\{L_{2}\left(M_{g,i j}, l_i, \hat{S}_{g,i j}\right) + {\hat{\alpha}}_{j}^{(t)} (M_{g,i j} - \hat{N}_{g,i j}^{(t)}) +\frac{\rho}{2} {(M_{g,i j} - \hat{N}_{g,i j}^{(t)})}^2\bigg\}.$$
		
		\underline{Step 2}. Update
		$$\hat{N}_{g,i \cdot}^{(t+1)} = {\left(\rho I + \hat{\Theta}_{g}\right)}^{-1} \left(\rho \hat{M}_{g,i \cdot}^{(t+1)} - \hat{\alpha}^{(t)} + \hat{\Theta}_{g} \hat{\mu}_g  \right).$$
		
		\underline{Step 3}. Update 
		$${\hat{\alpha}}^{(t+1)} = {\hat{\alpha}}^{(t)} + \rho \left(\hat{M}_{g,i \cdot}^{(t+1)} - \hat{N}_{g, i \cdot}^{(t+1)}\right).$$
		
		\underline{Evaluation step}. \ Update \begin{equation*}
			\begin{aligned}
				\delta_{M} &= \delta\left(L_M \left(\hat{M}_{g, i \cdot}^{(t+1)},\hat{N}_{g, i \cdot}^{(t+1)}\right),  L_M \left(\hat{M}_{g, i \cdot}^{(t)},\hat{N}_{g, i \cdot}^{(t)}\right)\right),\ \ \delta \left(a,b\right) = \left| a - b\right|/b, 
				\ t = t + 1.
			\end{aligned}
		\end{equation*} 
		\BlankLine
	}
\end{algorithm}

\subsection{Tuning Parameters Selection} 
We select the tuning parameter $\lambda_n > 0$ for each $\Theta_g$ independently by minimizing the integrated complete likelihood criterion (ICL) \citep{biernacki2000assessing}
\begin{equation}
	-2\ {l^{(g)}_{\text{E}}} \left(\hat{\eta},\hat{\theta} \right) + \log \left(1_{n}^{T} \hat{P}_{\cdot g}\right) s(\hat{\Theta}_{g}),
\end{equation}
where $s(\hat{\Theta}_{g})$ denotes the number of nonzero elements in $\hat{\Theta}_{g}$. We can also select the tuning parameter $\lambda_n$ such that the estimated network has a desired density.

\section{Simulation}
\label{sec:simu}

In this section, we perform simulation to evaluate the performance of VMPLN and compare with state-of-the-art graphical methods and single-cell gene regulatory network estimation methods, including Glasso \citep{friedman2008sparse}, LPGM \citep{allen2013local}, VPLN \citep{chiquet2019variational}, PPCOR \citep{kim2015ppcor}, GENIE3 \citep{huynh2010inferring} and  PIDC \citep{chan2017gene}. For VMPLN, we use clustering results given by the K-means algorithm as the initial value and infer the networks jointly for all populations. For the other algorithms, we cluster the samples using the K-means algorithm and infer network for each cluster separately. 

\subsection{Simulation Setups}
We first generate simulation data based on the mixture Poisson log-normal model. The number of populations is set as $G=3$ and the proportion parameter $\pi$ is set as $(1/3,1/3,1/3)$. The number of observations is $n=3000$. We consider 48 different simulation scenarios, which are 3 population-mixing levels (low, middle and high) $\times$ 2 dropout levels (i.e. the percent of zeros in data, low and high) $\times$  2 dimension setups ($p=100$ and $300$) $\times$ 4 graph structures. In each scenario, we generate 50 datasets, and the details of the data generation process are shown in appendix Section \ref{Simulation:datageneration}. The four graph structures are as follows.

1. Random Graph: Pairs of nodes are connected with probability 0.1. The nonzero edges are randomly set as 0.3 or -0.3.

2. Hub Graph: 20 $\%$ nodes are set as hub nodes. A hub node is connected with another node with probability 0.1. Non-hub nodes are not connected with each other. The nonzero edges are randomly set as 0.3 or -0.3.

3. Blocked random Graph: The nodes are divided into 5 blocks of equal sizes. Pairs of nodes within the same block are connected with probability 0.1.  Nodes in different blocks are not connected. The nonzero edges are randomly set as 0.3 or -0.3.

4. Scale-free Graph: The Barabasi-Albert model \citep{barabasi1999emergence} is used to generate a scale-free graph with power 1. The nonzero edges are randomly set as 0.3 or -0.3.
 
\subsection{Performance Comparison}
\begin{figure}[htbp]
	\centering
	\includegraphics[width = 1.02\textwidth, height=0.8\textwidth]{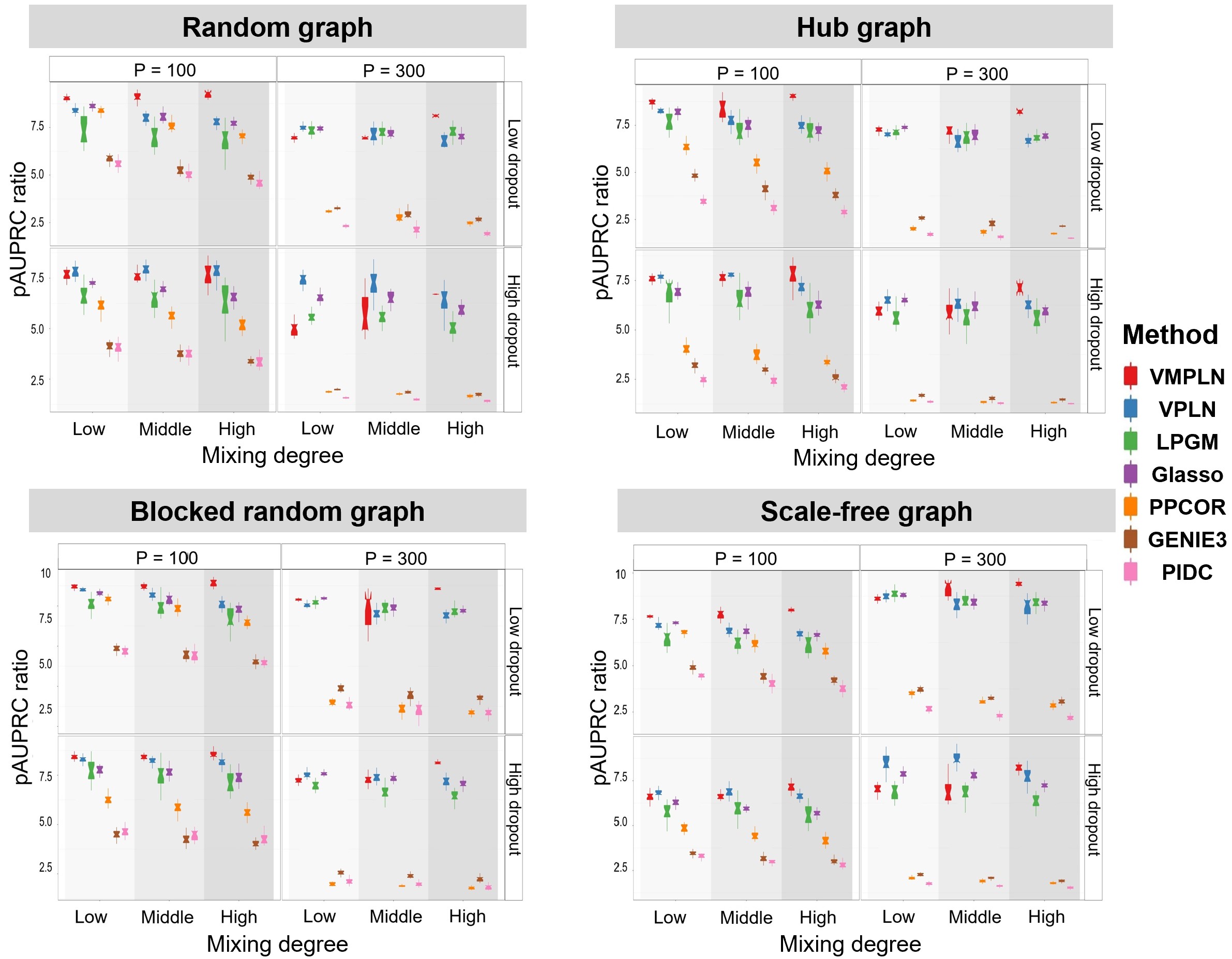}
	\caption{The pAUPRC ratios for four graphs. The parameters are set as their default values or are tuned by their default methods.}
	\label{AUPRC-EPR-1}
\end{figure}
The available methods often report network estimates with very different densities. Dense network predictions usually have a high sensitivity and a low specificity, but sparse network predictions have a low sensitivity and a high specificity. In order to make dense and sparse network predictions comparable, we adopt one criteria used in a previous benchmark work on network estimation methods \citep{pratapa2020benchmarking}, called the pAUPRC ratio, to evaluate the algorithms. Briefly, given a network estimation by an algorithm, we calculate its area under the partial precision-recall curve (called pAUPRC) by applying varying thresholds to the edge scores given by the algorithm (appendix Section \ref{Simulation:furtherint}). The pAUPRC ratio is defined as the ratio between the pAUPRC and the expected pAUPRC of the random network prediction of the same density.

We first compare the algorithms using their default parameters or default ways of selecting tuning parameters (appendix Section \ref{Simulation:furtherint}). Figure \ref{AUPRC-EPR-1} show the boxplots of pAUPRC ratios of different algorithms. Overall, VMPLN is the best-performing algorithm in terms of the pAUPCR ratio. The advantage of VMPLN is more pronounced when the populations have a higher mixing level, suggesting that compared with the two-step procedure, the joint analysis of network inference and clustering can help to improve the network inference. 
\begin{figure}[htbp]
	\centering
	\includegraphics[width = 1.02\textwidth, height=0.80\textwidth]{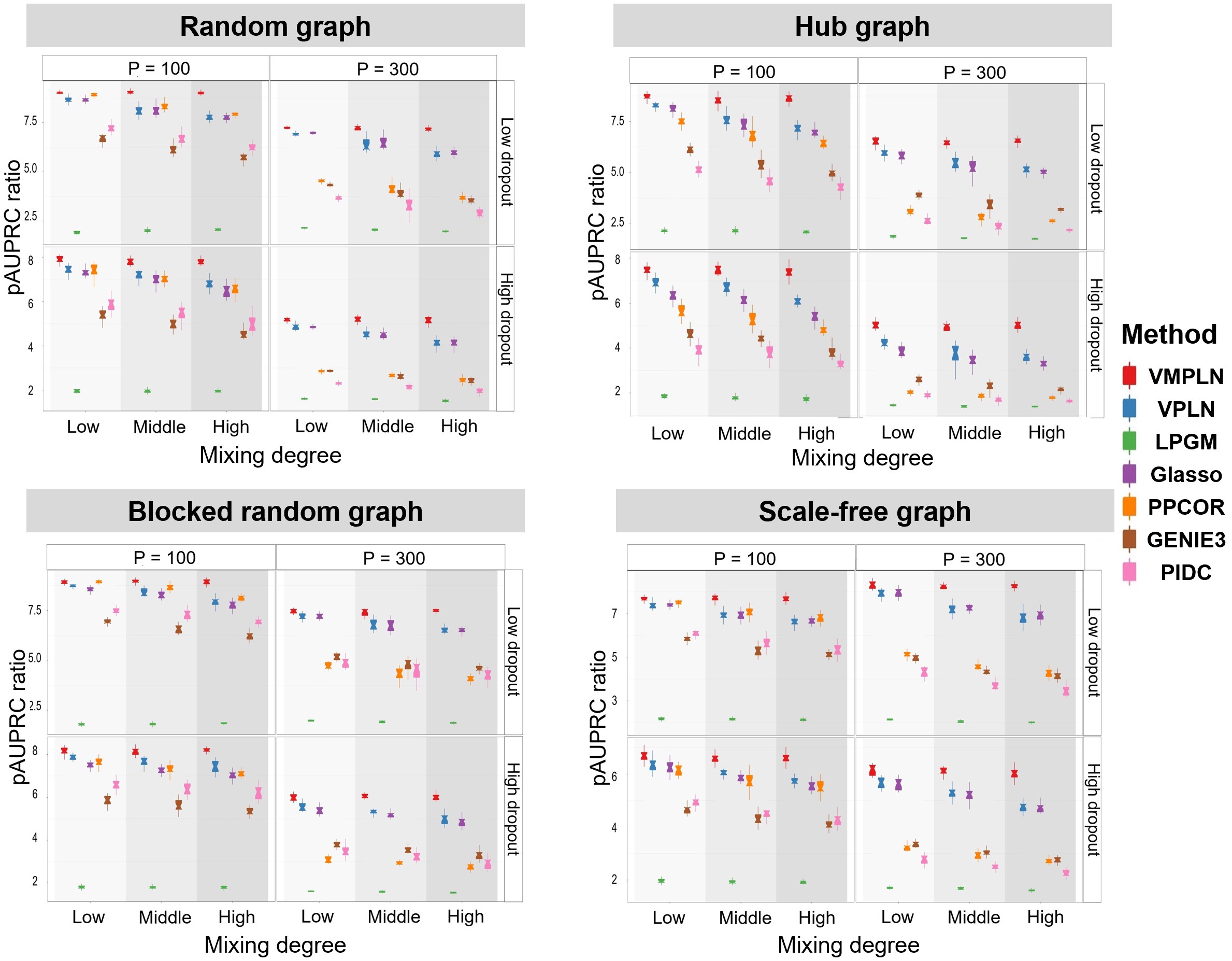}
	\caption{The pAUPRC ratios for four graphs. The edge score cutoffs or the tuning parameters are selected such that the network density is 20\%.}
	\label{AUPRC-EPR-2}
\end{figure}

Tuning parameter selection methods can have a large impact on the network estimates. To eliminate this influence, we tune the parameters such that the estimated network from different algorithms having 20\% nonzero edges (2 $\times$ the density of the true network) (appendix Section \ref{Simulation:furtherint}), and compare their pAUPRC ratios (Figure \ref{AUPRC-EPR-2}). Similarly, VMPLN has the best-performance in most scenarios, especially in the cases with high mixing levels. For example, in the simulation of the hub graph with $p=100$ and low-dropout, VMPLN has mean AUPRC ratios $8.63$ and $8.72$ in the high and low mixing scenarios, respectively, about $21\%$ and $5\%$ larger than the AUPRC ratios ($7.12$ and $8.26$) of VPLN under the same scenarios. All algorithms tend to have decreased performances in higher dimensions or with high dropout rates and VMPLN consistently has better performances in these more difficult settings.

\section{Real Data Analysis} 
\label{sec:appl}
\begin{figure}[htbp]
	\centering
	\includegraphics[width = 1.00\textwidth, height=0.86\textwidth]{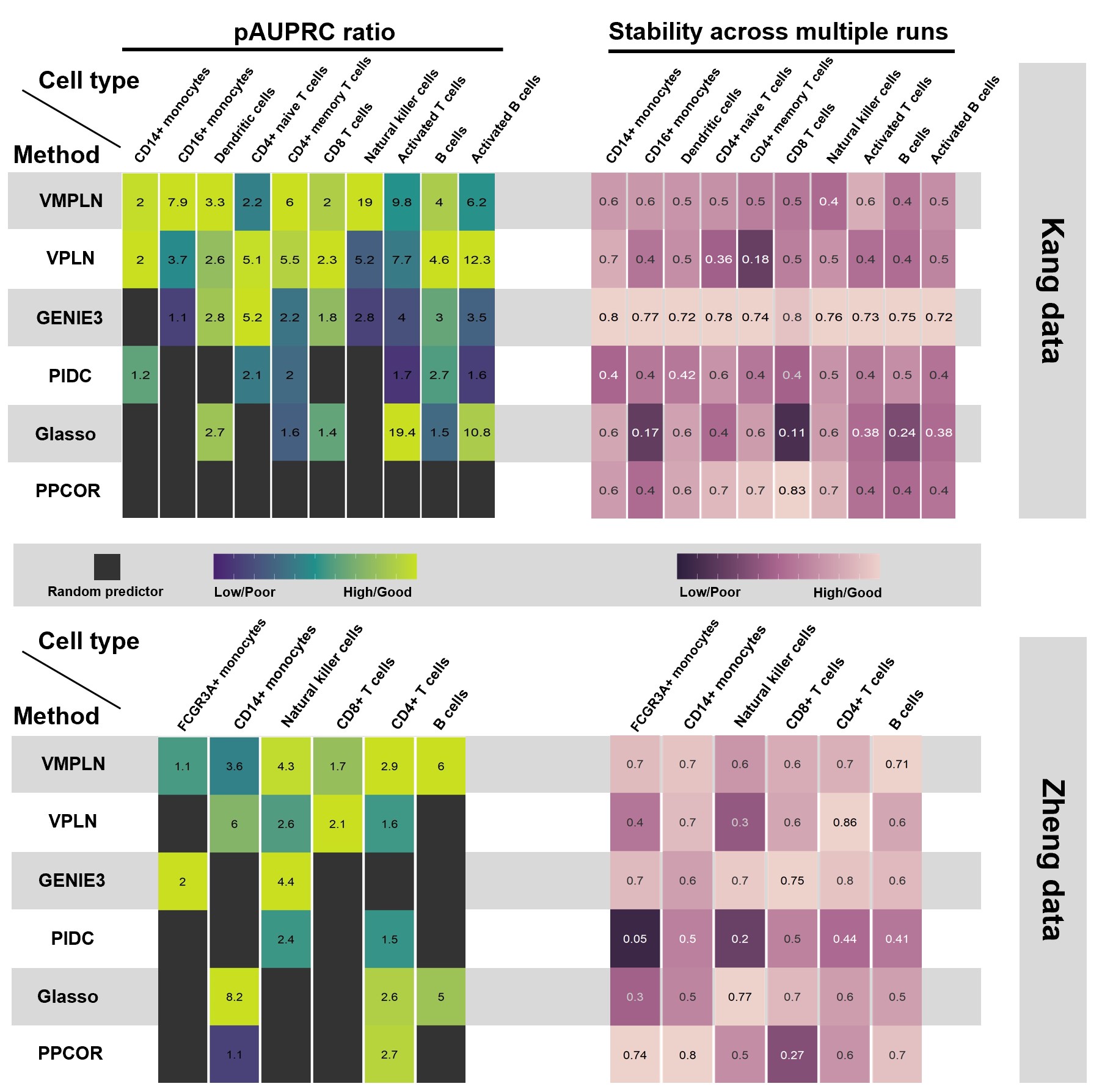}
	\caption{The performance of the netwok inference algorithms in two evaluation scRNA-seq datasets. The colors represent the scaled values of these metrics within each cell type and the actual values are marked in the boxes. Black color in the boxes: random predictor performs better.}
	\label{realdata_benchmark}
\end{figure}
In this section, we first evaluate the performance of VMPLN and compare with other methods using two real scRNA-seq datasets for gene regulatory network inference. Then, we demonstrate an application of VMPLN to a large scRNA-seq dataset from patients infected with SARS-CoV-2. The two evaluation datasets are scRNA-seq data profiled by \citet{kang2018multiplexed} (Kang data, 7217 cells from 10 cell types) and \citet{zheng2017massively} (Zheng data, 5962 cells from 6 cell types), respectively. Both datasets consist of two batches. We use one of the two batches and public gene regulatory network databases to construct silver standards (appendix Section \ref{Realdata:GRNdb} - \ref{Realdata:silverstandard}), and test different algorithms using another batches. Regulatory relationships are inferred for the top 500 highly variable genes given by Seurat \citep{stuart2019comprehensive} and evaluated by comparing with the silver standards. 

We compare the algorithms at the same network density ($5\%$) in terms of the pAUPRC ratio and the stability (Figure \ref{realdata_benchmark}). The pAUPRC ratios are calculated by comparing with the silver standard. The stability is defined as the median of pairwise Jaccard indexes between the networks estimated from 100 random down-sampled ($90\%$) datasets. VMPLN has the highest pAUPRC ratios in most cases and is the only method that consistently performs better than the random predictor. The stability of VMPLN is also reasonable and roughly similar to GENIE3, which is expected because GENIE3 uses data perturbation for network estimation and should have good stability.  

\begin{figure}[htbp]
	\centering
	\includegraphics[width = 1.10\textwidth, height=0.55\textwidth]{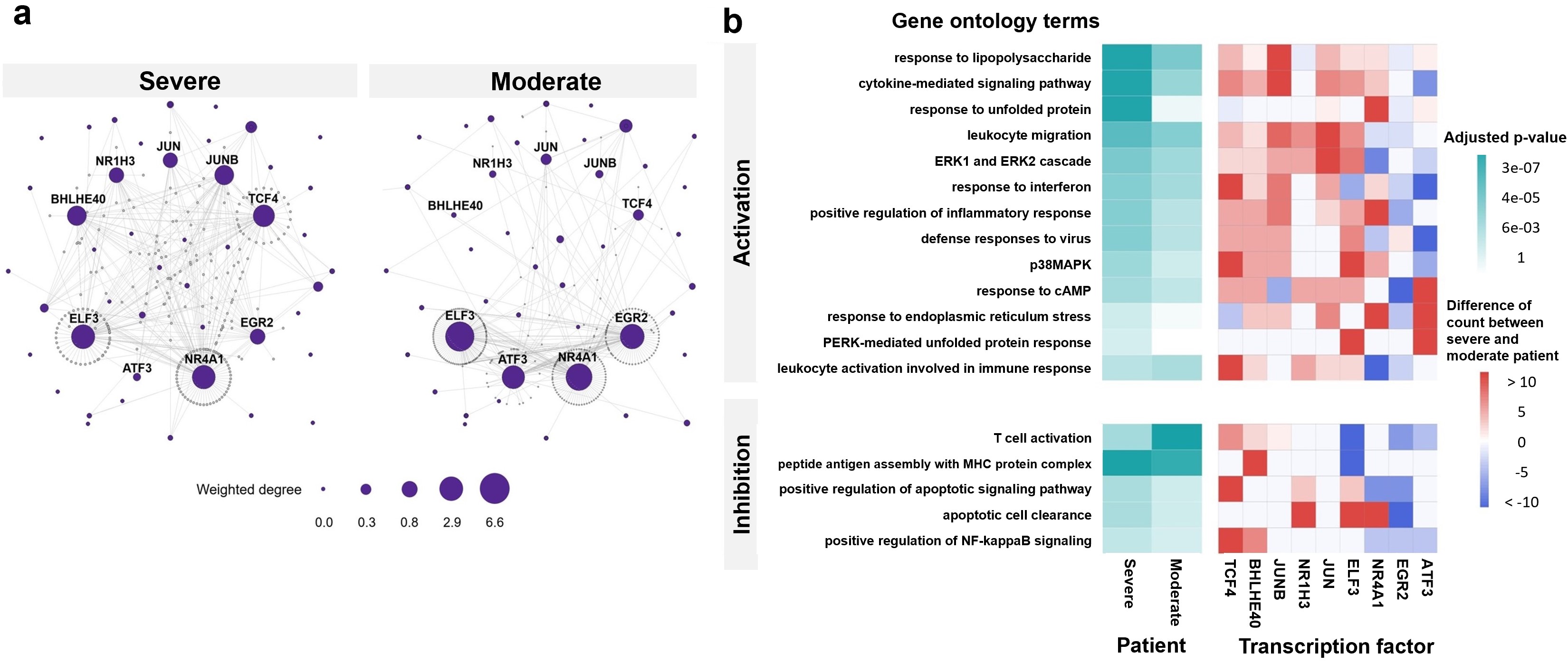}
	\caption{{The gene regulatory network analysis of the SARS-CoV-2 data}. (a) The inferred gene regulatory networks of Group4 macrophages in severe and moderate patients. The size of the node represents the weighted node degree. (b) Gene ontology enrichment analysis of Transcription factors' target genes. Transcription factors are selected as those with a large weighted degree differences ($> 0.2$) between severe and moderate patients. Activation and Inhibition mean that the regulatory relationships are positive and negative, respectively.  Left panel: p-values of gene ontology terms. Right panel: differences of number of genes in the gene ontology terms between severe and moderate patients.}
	\label{COVID_19_whole}
\end{figure}

We then apply VMPLN to a SARS-CoV-2 data dataset consisting of 29,980 bronchoalveolar lavage fluid macrophage cells from 8 patients, including 2 patients with SARS-CoV-2 infection and 6 patients with severe infection. \citet{liao2020single} clustered the macrophages to four clusters including two classic M1-like macrophage groups (Group1 and Group2), the alternative M2-like macrophages (Group3) and the alveolar macrophages (Group4). 
The gene regulatory network of the top 1000 highly variable genes are inferred using VMPLN  (Figure \ref{COVID_19_whole} a; appendix Section \ref{Realdata:SARSCOV2} and Figure \ref{GRN-group123}). We focus on the alveolar macrophages, since unlike other macrophage groups, the proportion of the alveolar macrophages tends to be smaller in patients with severe infection \citep{liao2020single}. 

A number of transcription factors exhibit a large weighted degree (weighted by absolute partial correlations) difference between the gene regulatory networks in the moderate and severe patients. 
Gene oncology enrichment analysis of their target genes (Figure \ref{COVID_19_whole} b) shows that, as expected, many of target genes are involved in immune responses such as leukocyte migration and regulation of T cell activation. Interestingly, we observe that a number of gene ontology terms are only enriched in severe patients such as response to unfolded protein and response (UPR) to endoplasmic reticulum (ER) stress, mainly due to the activation of target genes of NR4A1 (Figure \ref{COVID_19_whole} b). 
The UPR and ER stress processes are frequently activated in cells infected by viruses \citep{janssens2014emerging} including coronavirus \citep{chan2006modulation}, indicating that alveolar macrophages might be infected by SARS-CoV-2. In fact, a recent study showed that macrophages can be infected by SARS-CoV-2 and the infected macrophages activate T-cells to promote alveolitis in severe patients \citep{grant2021circuits}. In light of these findings, we reason that NR4A1 might play an important role in regulating cellular responses to SARS-CoV-2 infection.

\section{Discussion}
\label{sec:disc}
In this paper, we develop a regulatory network inference method called VMPLN for count data from mixed populations. Instead of using the two-step procedure for network inference, VMPLN performs clustering and network inference simultaneously, thus are better suitable for applications such as scRNA-seq studies. Since the Poisson log-normal model does not have a finite moment-generating function, common techniques for proving mixture models' identifiability and positive-definiteness of the Fisher information matrix do not work for the mixture Poisson log-normal model. We instead use all moments to prove these basic properties of mixture models. The techniques used here may also be useful for other mixture models.


\acks{This work was supported by the National Key Basic Research Project of China, the National Natural Science Foundation of China, and Sino-Russian Mathematics Center.}

\newpage

\appendix
\section{Appendix: Technical proofs.}
\label{Appendix:Technicalproofs}



\subsection{Notation}
\label{Appendix:Technicalproofs_notation}
For notational simplicity, we do not distinguish different lower (upper) bounds in Condition 1-3, and always use $m$ to denote the lower bound and $M$ to denote the upper bound. We always assume $M>0$. When $m$ is the lower bound for eigenvalues of precision matrices, we always assume $m>0$.  	We define two vectorization operators,  $\vecone$ and $\vectwo$. For a symmetric matrix $A \in \mathbb{R}^{p \times p} = [a_{ij}]$, $\vecone(A)$ is defined as
$$\vecone(A) = (a_{11}, a_{12}, a_{13},\ldots,a_{1p},  a_{22}, a_{23}, \ldots, a_{2p},\ldots,a_{(p-1)(p-1)}, a_{(p-1)p}, a_{pp})\trans,$$
and $\vectwo(A)$ is 
$$\vectwo(A) = (a_{11}, 2a_{12},2 a_{13},\ldots,2a_{1p},  a_{22}, 2a_{23}, \ldots, 2a_{2p},\ldots,a_{(p-1)(p-1)}, 2a_{(p-1)p}, a_{pp})\trans.$$
Observe that  $\vecone$ and $\vectwo$  only differ at off-diagonal elements. Define ${\nu} =  (\nu_1^{\rm T},\ldots,  \nu_g^{\rm T}, \ldots, \nu_G^{\rm T})^{\rm T},$ where  $\nu_g = {\rm vech}(\Theta_g)$. We assume the true parameter $\nu^*$  is an interior point of $\mathcal{D}$. 
For the Poisson log-normal model, we can write the log-likelihood as follows,
$$\log\, p(y,  \l;\Theta,  \mu) = \sum_{i=1}^{n} \log \int f(x;\Theta,  \mu)p(y_i\mid x, l_i) dx +\sum_{i=1}^{n}\log \left\lbrace p(l_i)\right\rbrace,$$
where $f(x; \Theta,  \mu)$ is the probability density function of ${\rm N}(\Theta, \mu)$. 
Observe that $\sum_{i=1}^{n}\log \left\lbrace p(l_i)\right\rbrace $ does not depend on $\Theta$.  We only need to consider the conditional distribution $p(y_i \mid l_i ;\Theta, \mu_g)$.
Since $p(l)$ is independent of the  unknown parameters and has a bounded support,  
it can be seen from the following section that the scaling factor $l$ does not have essential influence on the proof. For brevity, we assume that the scaling factor is the constant $1$ in the proof. 

In the following sections,  we always use  $p(y; \Theta, \mu)$ and  $p\left(y; \nu, \{\mu_g\}_{g=1}^G\right)$ to represent the  density of  the Poisson log-normal distribution ${\rm PLN }\left(\Theta, \mu\right)$ and  the  density of  the mixture Poisson log-normal distribution  MPLN $\left(\nu, \{\mu_g\}_{g=1}^G\right)$, respectively.
Given a single sample  $i$, we write the log-likelihood function  of the Poisson log-normal model at $y_i$ as 
\begin{align*}
	\ell(\Theta, y_i)& = \log\left\lbrace p(y_i;\Theta, \mu)\right\rbrace  = 2^{-1} \log \det(\Theta) \\ &  + \log 
	\int  \det(\Theta) ^{1/2} \exp\left\lbrace - 2^{-1} (x- \mu)\trans \Theta (x- \mu)\right\rbrace {h(y_i,x)}dx + C(y_i),
\end{align*}
where 
\begin{equation} \label{equ:h}
	h(y_i,x) = \prod_{j=1}^{p} \exp(x_{j} y_{ij}) \exp\left\lbrace -\exp(x_j)\right\rbrace 
\end{equation}
and $C(y_i) = \sum_{j=1}^p \log y_{ij}! - 2^{-1}p \log(2\pi)$. In the following sections, we always write $h(y,x) = \prod_{j=1}^{p} \exp(x_{j} y_{j}) \exp\left\lbrace -\exp(x_j)\right\rbrace $. Also, we define $\ell_n(\Theta) = \sum_{i=1}^{n} \ell(\Theta, y_i)$ as the log-likelihood in the Poisson log-normal model. 
For the mixture Poisson log-normal model,  its log-likelihood function at $y_i$ is
\begin{align*}
	\ell(\nu, y_i)& = \log \left\lbrace \sum_{g=1}^{G}\pi_g  p(y_i ;\Theta_g, \mu_g)\right\rbrace  \\ 
	&  = \log \left[\sum_{g=1}^{G}\pi_g  \int \det\left(\Theta_g\right) ^{1/2}\exp \left\lbrace - 2^{-1} (x- \mu_g)\trans \Theta_g (x- \mu_g)\right\rbrace    {h(y_i,x)}dx\right] + C(y_i),
\end{align*}
where  $C(y_i) = \sum_{j=1}^p \log y_j! - 2^{-1}p \log(2\pi)$. 
The log-likelihood of the mixture Poisson log-normal model is
${\ell}_n(\nu) =  \sum_{i=1}^{n}	\ell(\nu, y_i).$
If we define 
$$L_g(\nu_g, y) = \int \det(\Theta_g) ^{1/2}\exp \left\lbrace - 2^{-1} (x- \mu_g)\trans \Theta_g (x- \mu_g)\right\rbrace  
{h(y,x)}dx ,$$
and $L_M(\nu, y) = \sum_{g=1}^{G}\pi_g L_g(\nu_g, y)$, then ${\ell}(\nu, y_i)  = \log \left\lbrace L_M(\nu, y_i)\right\rbrace  + C(y_i)$.
Observe that the function $L_g(\nu_g, y)$  is proportional to the density 
$p(y ;\Theta_g, \mu_g)$.
Let 
$\mathcal{L}_n(\nu) = -n^{-1}\ell_n (\nu)$.
The optimization problem (3) in the main manuscript can be written as 
\begin{equation} \label{prob:opt}
	\hat{\nu}_{n} \in \argmin_{\nu \in {\mathcal{D}}} \left\{\mathcal{L}_n(\nu) + {\lambda_n} \mathcal{R}(\nu)\right\},
\end{equation}
where $\mathcal{R}(\nu) = \sum_{g=1}^G ||\Theta_g||_{1,\rm off}.$

For the Poisson log-normal model, denote the derivative (the score function) and the Hessian matrix of its log-likelihood as 
$$\mathcal{S}(\Theta, y) = \frac{\partial\ell(\Theta , y)}{\partial \vecone(\Theta)} , \quad H(\Theta, y) = \frac{\partial^2 \ell(\Theta, y)}{\partial \vecone(\Theta)\partial\vecone(\Theta)\trans} .$$ 
For the mixture Poisson log-normal model, we can similarly define its  score function $\mathcal{S}^M(\nu, y)$, its Hessian matrix $\F (\nu, y )$, and  its Fisher information matrix $\Gamma^*= \Gamma(\nu^*)$.    
\begin{equation} \label{equ:D}
	\D(\nu)=\Ex_{\nu^*} \left\{\F (\nu, y )\right\}. 
\end{equation}
Observe that $\D(\nu^*) = -\Gamma(\nu^*)$. 

Finally, we denote $\mathbb{N}^p$ as the set of all $p$-dimensional non-negative integer vectors.  For a vector ${a} = (a_1,\ldots,a_p)$, we denote $||{a}||_2 = \sqrt{\sum_{j=1}^p a_j^2}$ as its $L_2$-norm and $||{a}||_{\infty} = \max_{j} |a_j|$  as its $L_{\infty}$-norm. For a matrix $\bA$, we denote $||{A}||_2$ as its largest singular value of ${A}$ and $||{A}||_{1,\infty}$ as its maximum absolute row sum of ${A}$.  Given $\bTheta$ and $\bmu$, we define an operator $\mathcal{T}$ that maps functions in $x$ to functions in $y$,
$\mathcal{T}(f) =  \int \exp\left\lbrace -2^{-1} (\bx-\bmu)\trans \bTheta (\bx-\bmu) \right\rbrace f(\bx) h(\by,\bx)d\bx .$ We use $ \mathbb{I}(\bx) \equiv 1$ as the constant function taking value 1.

\subsection{Some Lemmas }
\begin{lemma}\label{lem:poly}
	Let $\by \sim {\rm PLN} (\bTheta, \bmu)$. 
	For any $n,y \in \mathbb{N}$,
	we define $	\phi(y,n) =    \prod_{i = 1}^n (y-i+1)$ if $n>0$ and $	\phi(y,n) = 1$ if $n = 0$.
	Then, for ${N} = (n_1,\dots, n_p)\trans$,
	we have 
	$
	\Ex\left(\prod_{j=1}^p \phi({n_j},y_j)\right) = \exp \left({{ N}}\trans \bmu + {{ N}}\trans \bTheta^{-1} {N}/2 \right).
	$  
\end{lemma}

\begin{lemma}[1-dimensional dominating function] \label{1dim}
	Suppose $\theta \in [m,M]$ ($m, M> 0$), and $|\mu| \leq M$.   Let 
	$$f^1(y,\theta,\mu) = {\int \exp \left\lbrace -2^{-1}\theta(x -  \mu)^2\right\rbrace  \exp\left\lbrace -\exp(x)\right\rbrace  \exp(xy)dx}.$$ 
	Then, for any large enough positive integer $y$, we have 
	$
	f^1(y,\theta,\mu)
	\geq  C \exp\left\lbrace y\log(y+1)/2\right\rbrace ,
	$
	where $C>0$ is  constant depending on $m,M$.
	
\end{lemma}

\begin{lemma}[$p$-dimensional dominating function] \label{pdim}
	Let 
	\begin{equation} \label{pdimequa}
		f^p(\by,\bTheta ,\bmu) = \frac{\mathcal{T}\left\{(x_{1} - \mu_1)^4\right\}}{\mathcal{T}(\mathbb{I})}, 
	\end{equation}
	where $\by=(y_1, \dots, y_p) \in \mathbb{N}^p$. Assuming $\by \sim {\rm PLN}(\bTheta^*, \bmu)$, under Condition 1-3, there exists a polynomial function $g(\by)= \left(4||\by||_2/m\right)^4 + C_p$ with a constant $C_p$ only depending on $p,m$ and $M$,
	such that  $\Ex_{\bTheta^*}(g(\by))<\infty$ and $|f^p(\by, \bTheta,\bmu)| \leq g(\by)$ for any $\bTheta, \bmu$ satisfying Condition 1-3.  
\end{lemma}
\begin{remark}\label{pdimremark}
	Applying the same proof as in Lemma \ref{pdim}, the polynomial $(x_{1} - \mu_1)^4 $ can be replaced by any polynomial with respect to $\bx$, e.g. $(x_{i} -\mu_i)^2(x_{j} - \mu_j)^2$ and $(x_{i} - \mu_i)(x_{j} - \mu_j)(x_{k} - \mu_k)^2$. Further, for any two polynomial functions ${\psi}_1 (\bx) , {\psi}_2 (\bx)$ and 
	$f^p(\by, \bTheta ,\bmu) = {\mathcal{T}\left\{{\psi}_1 (\bx)\right\} \mathcal{T}\left\{{\psi}_2 (\bx)\right\}}/{\mathcal{T}^2(\mathbb{I})}, 
	$ 
	there exists a polynomial function $g(\by)$ with $\Ex_{\Theta^*}(g(\by))<\infty$ such that  
	$|f^p(\by,\bTheta,\bmu)| \leq g(\by)$ for any $\bTheta,\bmu$ satisfying Condition 1-3.  
\end{remark}
\begin{proof}[of Lemma \ref{lem:poly}]
	By the property of conditional expectation, we have 
	$
	\Ex\left(\prod_{j=1}^p \phi({n_j},y_j)\right) = \Ex_{\bx}\Ex_{\by}\left(\prod_{j=1}^p \phi({n_j},y_j)\mid\bx\right). 
	$
	From the moments of the Poisson distribution, we have 
	$$\Ex_{\by}\left(\prod_{j=1}^p \phi({n_j},y_j)\mid\bx\right) = \prod_{j=1}^p \exp(n_j x_j) .$$
	Further, since $\bx \sim {\rm N}(\bmu, \bTheta^{-1})$, we have
	$$\Ex_{\bx}\left( \prod_{j=1}^p \exp(n_j x_j)\right) = \Ex_{\bx}\left\{\exp({N}\trans\bx)\right\} = \exp\left({N}\trans \bmu +  {N}\trans \bTheta^{-1} {N}/2\right),$$
	and the conclusion follows.
\end{proof}

\begin{proof}[of Lemma \ref{1dim}]
	Since for any $y \in \mathbb{N}$, $f^{1}(y, \theta, \mu) > 0$,  we only need to consider $y$ large enough.
	Let 
	$g(x,y) = \exp \left\lbrace -2^{-1}\theta(x - \mu)^2\right\rbrace  \exp\left\lbrace -\exp(x)\right\rbrace  \exp(xy).$
	Clearly, we have, for $y$ large enough,
	$\int g(x,y)dx 
	\geq \int_{\log(y)-1}^{\log(y)} g(x,y)dx.$
	Let $t = x-\log(y)$, then 
	\begin{align}\label{1dimlower}
		&\quad \int_{\log(y)-1}^{\log(y)} g(x,y)dx  \notag \\
		&= \exp\left\lbrace y\log(y)\right\rbrace \int_{-1}^{0} \exp\left[-2^{-1}\theta\left\lbrace \log(y)+t- \mu\right\rbrace ^2\right] \exp\left[-\exp\left\lbrace \log(y)+t\right\rbrace \right] \exp(yt)dt \notag\\
		&\geq \exp\left\lbrace y\log(y)\right\rbrace \min_{t\in[-1,0]} \left( \exp\left[ -2^{-1}\theta\left\lbrace \log(y)-t- \mu\right\rbrace ^2\right] \right) 
		\int_{-1}^0 \exp\left[ y\left\lbrace t-\exp(t)\right\rbrace \right] dt.\notag\\
		&\geq \exp\left\lbrace y\log(y)-y\right\rbrace \min_{t\in[-1,0]} \left( \exp\left[ -2^{-1}\theta\left\lbrace \log(y)-t- \mu\right\rbrace ^2\right] \right)
		\int_{-1}^0 \exp(-t^2/2)dt /\sqrt{y}. \notag\\
		&\geq \exp \left\lbrace y\log(y)-y\right\rbrace C_2 \exp \left\lbrace -C_3\log^2(y)+C_4\log(y)\right\rbrace /\sqrt{y},
	\end{align}
	where $C_2$, $C_3$ and $C_4$ are three constants only depending on $m,M$.
	Since the leading order of (\ref{1dimlower}) is $\exp \left\lbrace y\log(y)-y\right\rbrace  $ , when $y\geq 1$, there exists a constant $C$ depending on $m,M$ such that for any large enough positive integer $y$, we have 
	$	f^1(y,\theta,\mu) \geq C\exp \left\lbrace y\log(y+1)/2\right\rbrace ,$
	which proves the conclusions.
\end{proof}

\begin{proof}[of Lemma \ref{pdim}]
	Similarly, we only need to consider $||\by||_2$ large enough.  The proof consists of the following three steps.
	
	\noindent
	\textbf{ Step 1.}  We first give a lower bound for the denominator of (\ref{pdimequa}). 
	Observe that 
	$$\exp \left\lbrace -2^{-1}  (\bx - \bmu)\trans \bTheta  (\bx - \bmu) \right\rbrace  \geq \exp \left\lbrace -2^{-1} M (\bx - \bmu)\trans   (\bx - \bmu) \right\rbrace .$$ We have
	\begin{align*}
		&\quad \int \exp \left\lbrace -2^{-1}  (\bx - \bmu)\trans \bTheta  (\bx - \bmu) \right\rbrace h(\by,\bx)d\bx 
		\geq \int\exp \left\lbrace -2^{-1} M (\bx - \bmu)\trans   (\bx - \bmu) \right\rbrace h(\by,\bx)d\bx \\
		&= \prod_{j=1}^p \int \exp \left\lbrace -2^{-1} M(x_j - \mu_j)^2\right\rbrace \exp \left\lbrace -\exp(x_j)\right\rbrace  \exp(x_jy_j)dx_j.
	\end{align*}
	By Lemma \ref{1dim}, we have, for any fixed $j$, $ \int \exp \left\lbrace -2^{-1}  M(x_j -\mu_j)^2 \right\rbrace \exp \left\lbrace -\exp(x_j)\right\rbrace  \exp(x_jy_j)dx_j$ is greater than $C\exp \left\lbrace y_j\log(y_j + 1)/2 \right\rbrace $. Hence,  we have
	$\int \exp \left\lbrace -2^{-1} (\bx - \bmu)\trans \bTheta  (\bx - \bmu)\right\rbrace h(\by,\bx)d\bx \geq C^p  \exp \left\lbrace \sum_j y_j\log(y_j + 1)/2\right\rbrace .$
	
	\noindent
	\textbf{ Step 2.}  When $||\bx- \bmu||_2 \leq A(\by) = 4||\by||_2/m$, we have $(x_1 - \mu_1)^4 \leq ||\bx- \bmu||_2^4\leq A^4(\by)$, where $x_1$ is the first element of $\bx$. Then, we have 
	$$
	\frac{\int_{||\bx-\bmu||_2 \leq A} \exp \left\lbrace -2^{-1}  (\bx - \bmu)\trans \bTheta  (\bx - \bmu)\right\rbrace   (x_{1} - \mu_1)^4h(\by,\bx)d\bx}{\int \exp \left\lbrace -2^{-1}  (\bx - \bmu)\trans \bTheta  (\bx - \bmu) \right\rbrace h(\by,\bx)d\bx} \leq A^4(\by). 
	$$
	Since $A^4(\by)$ is a polynomial function of $\by$, we have $\Ex_{\bTheta^{*}}(A^4(\by)) < \infty.$
	
	\noindent
	\textbf{ Step 3.}  When $||\bx-\bmu||_2 > A(\by)$,  we have
	\begin{align*}
		&\quad \int_{||\bx-\bmu||_2 > A(\by)} \exp \left\lbrace -2^{-1}  (\bx - \bmu)\trans \bTheta  (\bx - \bmu)\right\rbrace (x_{1} -\mu_1)^4 h(\by,\bx)d\bx \\
		&\leq  \int_{||\bx - \bmu||_2 > A(\by)} \exp \left\lbrace -2^{-1}  (\bx - \bmu)\trans \bTheta  (\bx - \bmu)\right\rbrace  (x_{1} - \mu_1)^4 \exp\left(\sum_{j=1}^{p}x_{j}{y_{j}}\right) d\bx\\
		&= \prod_{j=1}^{p} \exp(\mu_j{y_{j}}) \int_{||{{u} }||_2 > A(\by)} \exp\left(-2^{-1}  {{u} }\trans \bTheta {{u} }\right) u_1^4\prod_{j=1}^{p} \exp(u_{j}{y_{j}}) d{{u} } \quad({{u} }= \bx - \bmu)\\
		&\leq \prod_{j=1}^{p} \exp(\mu_j{y_{j}})\int_{||{{u} }||_2 > A(\by)} \exp\left(-2^{-1} m ||{{u} }||_2^2\right) ||{{u} }||_2^4 \exp(||{{u} }||_2||\by||_2)d{{u} },
	\end{align*}
	where the last inequality is by Cauchy's inequality. Observe that the area of the $p$-dimensional sphere of radius $r$ is $M_p r^{p-1}$ with $M_p$ being a constant only depending on $p$. Let $r = ||{{u} }||_2$. By the polar decomposition, the $p$-dimensional integral can be rewritten as
	\begin{align*}
		&\quad \int_{||{{u} }||_2 > A(\by)} \exp\left( -2^{-1} m ||{{u} }||_2^2\right)  ||{{u} }||_2^4 \exp\left( ||{{u} }||_2||\by||_2\right) d{{u} } \\
		&= \int_{r>A(\by)} \exp \left\lbrace -2^{-1} r \left( m r - 2||\by||_2\right) \right\rbrace  M_p r^{p+3}dr \\
		&\leq \int_{r>A(\by)} \exp(-r ||\by||_2) M_p r^{p+3}dr  \quad \left( r>4\frac{||\by||_2}{m}\right) \\
		&\leq \int_{r>A(\by)} \exp(-r) M_p r^{p+3}dr = C'_p,
	\end{align*}
	where $C'_p$ is a constant  only depending on $p$. Then, under Condition 1-3, we have
	$$\frac{\int_{||\bx-\bmu||_2 \leq A} \exp \left\lbrace -2^{-1}  (\bx - \bmu)\trans \bTheta  (\bx - \bmu)\right\rbrace   (x_{1} - \mu_1)^4h(\by,\bx)d\bx}{\int \exp \left\lbrace -2^{-1}  (\bx - \bmu)\trans \bTheta  (\bx - \bmu) \right\rbrace h(\by,\bx)d\bx} \leq  \frac{ \prod_{j=1}^{p} \exp(\mu_j{y_{j}}) C'_p}{\exp \left\lbrace \sum_j y_j\log(y_j + 1)/2 \right\rbrace  } \leq C_p$$
	where $C_p$ is a constant depending on $p,m$ and $M$. 
	
	Finally, combining the results in Step 2 and 3,  we get the dominating function 
	$g(\by) = \left( 4||\by||_2/m\right) ^4 +C_p$
	which is a polynomial. By Lemma \ref{lem:poly}, we have $\Ex_{\bTheta^{*}}(g(\by)) < \infty$.
\end{proof}

\subsection{Proof of Theorem  1,  Part I}

To prove the the first conclusion of Theorem 1, we introduce the following definition and give two lemmas.

\begin{definition}[Good vector]
	We call a vector $\bxi = (\xi_1, \dots, \xi_G)\trans \in \mathbb{R}^G$ as a good vector if one $\xi_g$ only appears once  in $\bxi$, i.e. $\xi_{g'} \neq \xi_{g}$ for all $g' \neq g$. We call the index $g$ as a good index with respect to $\bxi$.
\end{definition}
\begin{lemma}\label{goodvec}
	Let $\bxi= (\xi_1, \dots, \xi_G)\trans $ be a good vector with a good index $s$,  $\bsigma = (\sigma_1,\dots, \sigma_G)\trans$ satisfy $\sigma_g >0$ for $g=1,\dots,G$ and $\balpha = (\alpha_1,\dots, \alpha_G)\trans$. If for any $z \in \mathbb{N}$,
	$\sum_{g=1}^{G}\alpha_g \exp(\xi_g z +  \sigma_g z^2/2) = 0,$
	then $\alpha_s = 0$.
\end{lemma} 
\begin{lemma} \label{lem:subspace}
	For any $n>0$, let $\mathcal{M}_i \subset \mathbb{R}^p, i=1,\dots, n$ be $n$ linear proper subspaces. Then, there  exists a non-negative integer vector $\bgamma$ such that $\bgamma \not\in \bigcup_{i=1}^n \mathcal{M}_i$.
\end{lemma}
\begin{proposition}\label{prop:density}
	$p\left(\by; \bTheta_1, \bmu_1\right), \dots, p\left(\by; \bTheta_G, \bmu_G\right)$ are linearly independent for $\bmu_g$ ($g=1,\dots, G$) that are bounded and different from each other. 
\end{proposition} 

\begin{proof}[of Theorem  1  part I]
	By \cite{yakowitz1968identifiability}, under Condition 1-3, the identifiability of the mixture Poisson log-normal model is equivalent to the linear independence of the Poisson log-normal components. Thus, we aim to prove   Proposition \ref{prop:density}.
	We prove this by mathematical induction.
	
	The independence for $G=1$ is trivial.  
	Now we assume that Proposition \ref{prop:density} holds for $G-1$. For any $\bmu_1,\dots, \bmu_G$ that are bounded and different from each other, if  we can prove that there exists $\balpha = (\alpha_1, \dots, \alpha_{G})\trans$ and an index $s \in \{1,\dots,G\}$ such that 
	$\sum_{g=1}^G \alpha_g p\left(\by;\bTheta_g, \bmu_g\right) = 0 \mbox{ and } \alpha_s=0,$
	then by induction, we have  $\balpha = 0$ and hence $p(\by;  \bTheta_g, \bmu_g)$ ($g=1,\dots,G$) are linearly independent. So our goal is to prove that if $\sum_{g=1}^G \alpha_g p(\by;\bTheta_g, \bmu_g) = 0$, then we can always find an index $s$ such that $\alpha_s = 0$.

	Let ${N} = (n_1,\dots, n_p)\trans$ be any non-negative integer vector.  Then, for any positive integer $z$,  by Lemma \ref{lem:poly},  there exists a polynomial function $p_z(\by) = \prod_{j=1}^p  \phi(z n_j,y_j), z\in \mathbb{N} $ such that $$\sum_{g=1}^G \alpha_g \Ex_g\left\{p_z(\by)\right\} = 0 \mbox{ and } \Ex_g\left\{p_z(\by)\right\} = \exp\left(z N\trans \bmu_g + z^2{N}\trans \bTheta^{-1}_g {N} /2\right),$$
	where $\Ex_g$ represents taking expectation with respect to  ${\rm PLN}\left(\bTheta_g, \bmu_g\right)$.  Let 
	$$\bxi = \left(N\trans \bmu_1, \dots,  N\trans \bmu_G\right) \ \mbox{and} \ \bsigma = \left({N}\trans \bTheta^{-1}_1 N/2 , \dots, {N}\trans \bTheta^{-1}_G {N}/2\right).$$ By  Lemma \ref{goodvec}, if  there exists an $N$ such that $\bxi$ is a good vector with good index $s$, then $\alpha_s = 0$ and we complete the proof. If, on the other hand,  $\bxi = \left(N\trans \bmu_1, \dots,  N\trans \bmu_G\right)$ is not a good vector for any non-negative integer vector $N$.  Therefore, for any $N$, there exists $s
	\neq 1$ such that $ N\trans \bmu_1 =  N\trans \bmu_s$. Thus, $N$ is the solution to the linear equation 
	${ x}\trans( \bmu_1 - \bmu_s) = 0$. We define $\mathcal{M}_g$ as the linear space consisting of solutions to the linear equation ${ x}\trans\left(\bmu_1 - \bmu_g\right) = 0$ ($g \neq 1$) and $\mathcal{M} = \cup_{g=2}^G \mathcal{M}_g$.  Thus, for any non-negative integer vector $N$, we have $N \in \mathcal{M}$. Since $\bmu_g$ are different form each other, then $\dim{\mathcal{M}_g} = p-1$ and $\mathcal{M}_g$ is a proper subspace of $\mathbb{R}^p$. This is contradictory to Lemma \ref{lem:subspace}.  So there exists an $N$ such that $\bxi = \left(N\trans \bmu_1, \dots,  N\trans \bmu_G\right)$ is  a good vector and hence for any $\bmu_g$ ($g=1,\dots,G$) that are different from each other, $p(\by;\bTheta_1, \bmu_1), \dots, p(\by; \bTheta_{G}, \bmu_{G})$ are linearly independent. 
\end{proof}	

\begin{proof}[of Lemma \ref{goodvec}]
	Without loss of generality,  we assume that $(\xi_g,\sigma_g )$ ($g=1,\dots,G$) are increasingly ordered  (first by $\xi$ then by $\sigma$). We say that $(\xi_g,\sigma_g)$ and $(\xi_s,\sigma_s)$ are equivalent if $(\xi_g,\sigma_g) = (\xi_s,\sigma_s)$. By this equivalence relationship, 
	$\{(\xi_g,\sigma_g ) \}_{g=1}^G$ can be partitioned into $Q$ groups ($Q\geq 1$). Let $S_q$ be the index set of the $q$-th group. We have 
	$\sum_{q=1}^{Q} \sum_{j \in S_q} \alpha_j \exp(\xi_j z +  \sigma_j z^2/2) = 0$
	for all $z \in \mathbb{N}$.  
	Dividing $\exp(\xi_{G}z +  \sigma_{G} z^2/2) $  on both sides of the above equation, we get
	\begin{equation}\label{equ3}
		\sum_{q=1}^{Q-1} \sum_{j \in S_q} \alpha_j\exp(\xi_j z + \sigma_j z^2/2 - \xi_{G}z-  \sigma_{G} z^2/2) + \sum_{j \in S_{Q}} \alpha_j = 0
	\end{equation}  
	for all $z \in \mathbb{N}$.  By the choice of $\sigma_{G}, \xi_{G} $, the first summation of  (\ref{equ3}) converges to zero when $z$ goes to infinity. So, we have $\sum_{j \in S_{Q}} \alpha_j = 0$. By  mathematical induction, we have 
	$\sum_{j \in S_q} \alpha_j  = 0$ for $q = 1,\dots, Q$.  Since $\bxi$ is a good vector with a good index $s$, ($\xi_s,\sigma_s$) itself forms a group, and hence $\alpha_s = 0$.  
\end{proof}

\begin{proof}[of Lemma \ref{lem:subspace}]
	We prove by mathematical induction. The conclusion clearly holds for $n=1$.  Now we assume that Lemma \ref{lem:subspace}  holds for $n$ and we aim to prove that it also holds for $n+1$. 
	
	By induction hypothesis, we can take $\balpha \in \mathbb{N}^p	\setminus \bigcup_{i=1}^n \mathcal{M}_i$. If $\balpha \not\in \mathcal{M}_{n+1}$, we have $\balpha  \not\in \bigcup_{i=1}^{n+1} \mathcal{M}_i$, and the proof is finished.  Thus, we only need to consider  $\balpha \in \mathcal{M}_{n+1}$. Similarly, we can take ${{\beta}} \in \mathbb{N}^p	\setminus \bigcup_{i=2}^{n+1} \mathcal{M}_i$ and ${{\beta}} \in \mathcal{M}_1$.  
	For any $i\neq 1$, we can prove that  there is at most one $k_1$ such that $\balpha+k_1{\beta}\in  \mathcal{M}_i$. 
	In fact, if there are $k_1, k_2$ such that $k_1 \neq k_2$ and $\balpha + k _1{{\beta}}  \in  \mathcal{M}_i, \balpha + k _2{{\beta}}  \in  \mathcal{M}_i$, then 
	${{\beta}} \in   \mathcal{M}_i$, which is contradictory to the fact that ${{\beta}} \in \mathbb{N}^p	\setminus \bigcup_{i=2}^{n+1} \mathcal{M}_i$.   Furthermore, there is no $k\in \mathbb{N}$ such that $\balpha+k{{\beta}}\in \mathcal{M}_1$. If otherwise, there exists a  $k \in \mathbb{N}$ such that $\balpha + k {{\beta}}  \in  \mathcal{M}_1$. Then, we have
	$\balpha \in  \mathcal{M}_1$, which is also a contradiction. So we could find at most $n$ positive integers $k$  such that $\balpha + k {{\beta}} \in  \bigcup_{i=1}^{n+1} \mathcal{M}_i$. Since 
	there are infinitely many non-negative numbers, 
	we prove that
	there exists $k \in \mathbb{N}$ such that $\balpha + k {{\beta}} \not\in  \bigcup_{i=1}^{n+1} \mathcal{M}_i$, and Lemma \ref{lem:subspace}  is proved. 
\end{proof}

\subsection{Proof of Theorem  1,  Part II}
To prove this conclusion, we need to give the explicit formula for the score functions and the Fisher information matrices of the Poisson log-normal model  and mixture Poisson log-normal model.
The Hessian matrix $\bH\left(\bTheta, \by\right)$  of the Poisson log-normal model is a $p(p+1)/2 \times p(p+1)/2$ matrix. For  notational convenience, we let $\bH_{(i,j)(i',j')} = \partial^2 \ell(\bTheta, \by)/(\partial \Theta_{i'j'}\partial \Theta_{ij})$, $(i \leq j, i' \leq j')$ as the element at the $(2p-i+1)i/2 - p + j$ row and $(2p-i'+1)i'/2 - p + j'$  column of the Hessian matrix. 
It is clear that the densities of the Poisson log-normal distribution and mixture Poisson log-normal distribution satisfy the regularity conditions in \cite{shao2003mathematical}. Then, we can calculate the score function and the Fisher information as follows.
The score function  of the Poisson log-normal distribution can be   written as  
\begin{align*}
	\mathcal{S}(\bTheta, \by) = & 2^{-1} \vectwo\left( \bTheta^{-1}\right) \\
	& - 2^{-1} \frac{\int \exp \left\lbrace -2^{-1} (\bx-\bmu)\trans \bTheta (\bx-\bmu)\right\rbrace  \vectwo \left\lbrace (\bx-\bmu)(\bx-\bmu)\trans\right\rbrace h(\by,\bx)d\bx}{\int \exp \left\lbrace -2^{-1} (\bx-\bmu)\trans \bTheta (\bx-\bmu)\right\rbrace h(\by,\bx)d\bx}.
\end{align*}
Especially,  at the true parameter  $\bTheta^*$ , we have
$$\mathcal{S}(\bTheta^*,  \by) = 2^{-1} \Ex_{\bx}\left\{\vectwo\left(\bTheta^{*-1} - (\bx-\bmu)(\bx-\bmu)\trans \right) \mid \by\right\}.$$ 
We use $(i,j), i \leq j$ to index $\mathcal{S}$. Using the operator $\mathcal{T}$,  the element of the score function at $(2p-i+1)i/2 - p + j$ can be rewritten as
$\mathcal{S}_{(i,j)}\left(\bTheta,  \by\right) =2^{-1} \vectwo\left(\bTheta^{-1}\right)_{(i,j)} - 2^{-1} \frac{ \mathcal{T}\left[\vectwo \left\lbrace (\bx-\bmu)(\bx-\bmu)\trans\right\rbrace _{(i,j)}\right]}{\mathcal{T}\left(\mathbb{I}\right)}.$
Using the operator $\mathcal{T}$, the Fisher information matrix can be written as follows.  Let $\Sigma = \bTheta^{-1}.$
When $i=j, i'=j'$, 
\begin{align*}
	\bH_{(i,i)(i',i')}\left(\bTheta, \by\right) &=  -2^{-1}\Sigma_{ii'}\Sigma_{i'i} + \frac{1}{4} \frac{\mathcal{T} \left\lbrace {({\bx - \bmu})}_{i}^2({\bx - \bmu})_{i'}^2\right\rbrace  }{ \mathcal{T}\left(\mathbb{I}\right)} \\
	& - \frac{1}{4} \frac{ \mathcal{T} \left\lbrace ({\bx - \bmu})_{i'}^2\right\rbrace  \mathcal{T} \left\lbrace {({\bx - \bmu})}_{i}^2\right\rbrace  }{ \mathcal{T}^2\left(\mathbb{I}\right)}.
\end{align*}
When $i\neq j, i'=j'$, 
\begin{align*}
	\bH_{(i,j)(i',i')}\left(\bTheta, \by\right) &=  -\Sigma_{ii'}\Sigma_{i'j} + 2^{-1}  \frac{ \mathcal{T} \left\lbrace {({\bx - \bmu})}_{i}{({\bx - \bmu})}_{j}({\bx - \bmu})_{i'}^2\right\rbrace }{ \mathcal{T}\left(\mathbb{I}\right)} \\
	&- 2^{-1} \frac{ \mathcal{T} \left\lbrace ({\bx - \bmu})_{i'}^2\right\rbrace \mathcal{T} \left\lbrace {({\bx - \bmu})}_{i}{({\bx - \bmu})}_{j}\right\rbrace }{ \mathcal{T}^2\left(\mathbb{I}\right)}.
\end{align*}
When $i=j, i'\neq j'$, 
\begin{align*}
	\bH_{(i,i)(i',j')}\left(\bTheta, \by\right) &=  -\Sigma_{ii'}\Sigma_{j'i} + 
	2^{-1}  \frac{\mathcal{T} \left\lbrace {({\bx - \bmu})}_{i'}{({\bx - \bmu})}_{j'}({\bx - \bmu})_{i}^2\right\rbrace  }{ \mathcal{T}\left(\mathbb{I}\right)} \\
	&- 2^{-1} \frac{ \mathcal{T} \left\lbrace ({\bx - \bmu})_{i}^2\right\rbrace  \mathcal{T} \left\lbrace {({\bx - \bmu})}_{i'}{({\bx - \bmu})}_{j'}\right\rbrace  }{ \mathcal{T}^2(\mathbb{I})}.
\end{align*}
When $i\neq j, i'\neq j'$, 
\begin{align*}
	\bH_{(i,j)(i',j')}\left(\bTheta, \by\right) &= -(\Sigma_{ii'}\Sigma_{j'j} + \Sigma_{ij'}\Sigma_{i'j}) + 
	\frac{ \mathcal{T} \left\lbrace {({\bx - \bmu})}_{i'}{({\bx - \bmu})}_{j'}({\bx - \bmu})_{i}({\bx - \bmu})_{j}\right\rbrace  }{ \mathcal{T}\left(\mathbb{I}\right)} \\
	&- \frac{ \mathcal{T} \left\lbrace ({\bx - \bmu})_{i}({\bx - \bmu})_{j}\right\rbrace  \mathcal{T} \left\lbrace {({\bx - \bmu})}_{i'}{({\bx - \bmu})}_{j'}\right\rbrace }{ \mathcal{T}^2\left(\mathbb{I}\right)}.
\end{align*}
\begin{lemma}\label{lem:PLNFisher}
	Assume $\by \sim {\rm PLN}(\bTheta, \bmu)$. Under Condition 1-3,  there exist two polynomial functions $K_1(\by),K_2(\by)$  with $\Ex \left( K_1(y)\right) <\infty$ and $\Ex\left( K_2(y)\right) <\infty$ such that   for any $i, j, i^\prime, j^\prime$, $|\mathcal{S}_{(i,j)} (\bTheta, \by )|\leq K_1(\by)$, $|\bH_{(i,i)(i',i')}(\bTheta,\by)| \leq K_2(\by)$.
	
\end{lemma}
Now we consider the score function and the Fisher information matrix of the mixture Poisson log-normal distribution.  The score function  of the mixture Poisson log-normal distribution can be written as
$$\mathcal{S}^M(\bnu, \by) = L_M(\bnu, \by)^{-1}\left({\pi_1 \frac{\partial L_1(\bnu_1, \by)}{\partial\bnu_1}}, \ldots,{\pi_G \frac{\partial L_G(\bnu_G, \by)}{\partial\bnu_G}}\right) = \left\{\mathcal{S}^M_1(\bnu,\by), \dots, \mathcal{S}^M_G(\bnu, \by)\right\}.$$
Similarly, we use $(g,i,j), i\leq j$ to index $\mathcal{S}^M$.
Recall that $\bnu_g= \vecone(\bTheta_{g})$.  We use $(g,i,j),(g',i',j'), i\leq j, i'\leq j'$ to index the  element at the $(g-1)p(p+1)/2 + (2p-i+1)i/2 - p + j$ row and $(g'-1)p(p+1)/2 +(2p-i'+1)i'/2 - p + j'$  column of $F(\bnu, \by)$, respectively.  

\noindent
When $g = g'$
\begin{align}
	\F_{(g,i,j,g,i',j')}(\bnu,\by) &= \frac{\partial}{{\partial\bTheta_{g,i'j'}}}\mathcal{S}^M_{(g,i,j)}(\bnu, \by) \notag \\
	&= \frac{\pi_g L_g(\bTheta_g, \by)}{L_M(\bnu, \by)}  \left\lbrace   H_{(i,j)(i',j')}(\bTheta_g,\by) +  \mathcal{S}_{(i',j')}(\bTheta_g, \by	) \mathcal{S}_{(i,j)}(\bTheta_g, \by) \right\rbrace \notag\\
	&- \frac{\left\lbrace \pi_gL_g(\bTheta_g, \by)\right\rbrace ^2}{L_M(\bnu, \by)^2} 
	\mathcal{S}_{(i',j')}(\bTheta_g, \by	) \mathcal{S}_{(i,j)}(\bTheta_g, \by).
\end{align}
When $g \neq g'$, we have
\begin{align}
	\F_{(g,i,j,g',i',j')}(\bnu,\by) &= \frac{\partial}{{\partial\bTheta_{g',i'j'}}}\mathcal{S}^M_{(g,i,j)}(\bnu, \by)\notag \\
	&=- \frac{\pi_g\pi_g'L_g(\bTheta_g, \by)L_{g'}(\bTheta_{g'}, \by)}{L_M(\bnu, \by)^2} 
	\mathcal{S}_{(i',j')}(\bTheta_{g'}, \by	) \mathcal{S}_{(i,j)}(\bTheta_g, \by).
\end{align}
Recall the definition (\ref{equ:D}) of $\D(\bnu)$. Then, 
$\bGamma(\bnu^*) = -\D(\bnu^*) = \Ex \left\{\mathcal{S}^M(\bnu^*, \by) \mathcal{S}^M(\bnu^*, \by)\trans\right\}$
is the Fisher information matrix  of the mixture Poisson log-normal distribution at $\bnu^*$.
\begin{lemma}\label{lem:MPLNFisher}
	Assume $\by \sim {\rm MPLN}(\bnu, \bmu)$. Under Condition 1-3,  there exists a polynomial function $K(\by)$  with $\Ex\left( K(y)\right) <\infty$ such that   for any $g,i, j, g^\prime, i^\prime, j^\prime$, 	$|\F_{(g,i,j,g',i',j')}(\bnu,\by)| \leq K(\by)$.
	
\end{lemma}
In addition, we require the  following two lemmas.
\begin{lemma}\label{lem:tr}
	Let $\by$ and $\bx$ be random variables as in the Poisson log-normal model (1) in the main manuscript of the paper and $\phi(y,n)$ is the same as Lemma \ref{lem:poly}.
	Let ${{N}} = (n_1,\dots, n_p)\trans$ and $\bT$ be a $p\times p$ matrix.  $\bT \in \mathbb{R}^{p\times p}$.	We have 
	\begin{align*}
		&\quad \Ex\left(\prod_{j=1}^p \phi({n_j},y_j) \tr\left[\bT \left\lbrace \bTheta^{-1} -  ({\bx} -\bmu) ({\bx} -\bmu)\trans \right\rbrace \right] \right) \\ 
		&=\left({{N}}\trans \bTheta^{-1}\bT \bTheta^{-1} {{N}}\right) \exp\left({{N}}\trans \bmu + {{N}}\trans \bTheta^{-1} {{N}}/2\right).
	\end{align*}
	
\end{lemma}

\begin{lemma} \label{lem:mixsubspace}
	For any $n>0$, let $\mathcal{M}_i \subset \mathbb{R}^p, i=1,\dots, n$ be $n$ linear proper subspaces. Let $A \neq 0$ be a $p\times p$ symmetric matrix and $\mathcal{V} = \{\by: \by\trans A\by = 0 \}$.  
	Then, there  exists a non-negative integer vector $\bgamma$ such that $\bgamma \not\in \bigcup_{i=1}^n \mathcal{M}_i$ and $\bgamma \not\in \mathcal{V}$. 
\end{lemma}

\begin{proof}[of Theorem  1,  part II]
	Observe that 
	$\bGamma(\bnu^*) = \Ex_{\by}\left\{\mathcal{S}_M(\bnu^*, \by) \mathcal{S}_M(\bnu^*, \by)\trans\right\}.$
	If there exists a non-zero vector $\bt = (\bt_1,\dots, \bt_G)\trans$ such that 
	$\Ex(\bt \trans \mathcal{S}^M(\bnu^*, \by) \mathcal{S}^M(\bnu^*, \by)\trans \bt) = 0$, then we aim to prove that $\bt = 0$.  
	Since y is a discrete random variable,  it follows that for any $\by$, $\bt \trans \mathcal{S}^M(\bnu^*, \by) = 0$.   Then, we have 
	${L_M(\bnu, \by)} ^{-1}\sum_{g=1}^{G} \bt_g  {\pi_g \frac{\partial L_g(\bnu_g, \by)}{\partial\bnu_g}}= 0.$
	Since $L_M(\bnu, \by) \neq 0$, we have 
	$$\sum_{g=1}^{G} \bt_g \trans{\pi_g \frac{\partial L_g(\bnu_g, \by)}{\partial\bnu_g}} = 0.$$
	Let $\psi(\by) = \prod_{j=1}^p  \phi(z n_j,y_j), z\in \mathbb{N}^p$.  Then, we have 
	$\sum_{g=1}^{G} \bt_g \trans{\pi_g \frac{\partial \log L_g(\bnu_g, \by)}{\partial\bnu_g}} L_g(\bnu_g, \by)\psi(\by)= 0.$
	Since $L_r(\bnu_r, \by)$ is proportional to the density of the Poisson log-normal model  with parameters $\bTheta_g, \bmu_g$, 
	the above equation can be rewritten as 
	$\sum_{g=1}^{G} \bt_g \trans{\pi_g \frac{\partial \log L_g(\bnu_g, \by)}{\partial\bnu_g}} p({\by} ; \bTheta_g, \bmu_g) \psi(\by)= 0.$
	Summing over $\by$, we get 
	$\sum_\by \sum_{g=1}^{G} \bt_g \trans{\pi_g \frac{\partial \log L_g(\bnu_g, \by)}{\partial\bnu_g}} p({\by} ; \bTheta_g, \bmu_g) \psi(\by)= 0.$
	By Fubini 's Theorem, we get
	
	\begin{equation}\label{prop2.1}
		\sum_{g=1}^{G}\sum_{\by} \bt_g \trans{\pi_g \frac{\partial \log L_g(\bnu_g, \by)}{\partial\bnu_g}} p({\by} ; \bTheta_g, \bmu_g) \psi(\by)= 0. 
	\end{equation}
	Then, let ${N} =  (n_1,\dots, n_p)\trans$ and  $\bT_g$ be the symmetric matrix  such that $\vecone(\bT_g) = \bt_g$. For a fixed $g$, we have 
	\begin{align*}
		&\quad \sum_{{\by}} \bt_g \trans{ \frac{\partial \log L_g(\bnu_g, {\by} )}{\partial\bnu_g}} p({\by} ; \bTheta_g, \bmu_g)\psi({\by} ) = \Ex_{{\by}_g} \left\{\psi(\by_g) \bt_g \trans{ \frac{\partial \log L_g(\bnu_g, \by_g )}{\partial\bnu_g}}\right\} \\
		&=2^{-1} \Ex_{\bx_g, \by_g} \left[\psi(\by_g) \tr\left\lbrace \bT_g\bTheta^{*-1}_g - \bT_g (\bx_g -\bmu) (\bx_g -\bmu)\trans\right\rbrace\right],
	\end{align*}
	where $\by_g$ follows the Poisson log-normal distribution with parameters $\bTheta_g$ and $\bmu_g$,  $\bx_g \sim {\rm N}\left(\bmu_g, \bTheta_g^{-1}\right)$ is the corresponding latent variable.
	By Lemma \ref{lem:tr}, we get 
	$$\Ex_{\bx_g, \by_g}\left[\psi(\by_g) \bt_g \trans{ \frac{\partial \log L_g(\bnu_g, \by_g)}{\partial\bnu_g}}\right] = z^2\left({{N}}\trans \bTheta^{-1}_g\bT_g  \bTheta^{-1}_g {{N}}\right) \exp\left(z{{N}}\trans \bmu_g +  z^2{{N}}\trans \bTheta_g^{-1} {{N}}/2\right).$$
	Then, (\ref{prop2.1}) can be rewritten as, for all $z \in\mathbb{N}$, 
	$$ \sum_{g=1}^{G}\pi_gz^2 \left({{N}}\trans \bTheta^{-1}_g\bT_g  \bTheta^{-1}_g {{N}}\right) \exp\left(z{{N}}\trans \bmu_g+  z^2{{N}}\trans \bTheta_g^{-1} {{N}}/2\right)= 0.$$
	In order to show $\bT_1 = 0$,  similar to the proof of the first conclusion,  we define $\mathcal{M}_g$ as   linear space consisting of  solutions to the linear equation ${ x}\trans( \bmu_1 - \bmu_g) = 0$ ($g= 2, \dots, G$) and $\mathcal{M} = \cup_{g=2}^G \mathcal{M}_g$. For any $N \not\in \mathcal{M} $, then $(N\trans \bmu_1, \ldots, N\trans \bmu_G)$ is a good vector with a good index $1$.  Since $\pi_1 > 0$,  we must have ${{N}}\trans \bTheta^{-1}_1\bT_1  \bTheta^{-1}_1 {{N}} = 0$.  By Lemma \ref{lem:mixsubspace},  if $\bTheta^{-1}_1 \bT_1  \bTheta^{-1}_1$ is not a zero matrix, then there exists an ${{N}}$ such that $N \not\in \mathcal{M} $ and ${{N}}\trans \bTheta^{-1}_1\bT_1  \bTheta^{-1}_1 {{N}} \neq 0$, which is contradictory to the fact that $N \not\in \mathcal{M} $ implies  ${{N}}\trans \bTheta^{-1}_1\bT_1  \bTheta^{-1}_1 {{N}} = 0$. Hence, we must have $\bTheta^{-1}_1 \bT_1  \bTheta^{-1}_1 = {0}$ and thus $\bT_1 ={0}$.  Similarly, we get $\bT_g ={0}$ for all $g = 1,\ldots, G$.  It follows that $\bt  = {0}$, and we compete the proof. 
\end{proof}

\begin{proof}[of Lemma \ref{lem:PLNFisher}]
	We only prove that there is a dominating function for $\mathcal{S}_{(i,j)}(\bTheta,\by)$. Others can be proved similarly.  Since  $m \leq \lambda_{min}(\bTheta)\leq\lambda_{max}(\bTheta)\leq M$,  we have   $\lambda_{max}(\bSigma)\leq 1/m$ and thus 
	$\vectwo(\bTheta^{-1})_{(i,j)}$ is bounded by a constant. 
	Further, since
	$\vectwo \left\lbrace (\bx-\bmu)(\bx-\bmu)\trans\right\rbrace _{(i,j)}$ is a polynomial function of $\bx$,   by  Remark \ref{pdimremark}, we have 
	$2^{-1} { \mathcal{T}\left[\vectwo \left\lbrace (\bx-\bmu)(\bx-\bmu)\trans \right\rbrace_{(i,j)}\right]}\big /{\mathcal{T}\left(\mathbb{I}\right)}$
	can be bounded by an integrable  polynomial function and we prove the  existence of $K_1(\by)$. 
\end{proof}

\begin{proof}[of Lemma \ref{lem:MPLNFisher}]
	If $g = g'$, since $ \pi_gL_g(\bnu, \by) / L_M(\bnu, \by)  \leq 1 $,  then we have, 	
	$$
	|\F_{(g,i,j,g,i',j')}(\bnu, \by)|  \leq    |\bH_{(i,j,i'j')}(\bTheta_g, \by )|  + 2|\mathcal{S}_{(i,j)}(\bTheta_g, \by	)\mathcal{S}_{(i',j')}(\bTheta_g, \by	)|.
	$$
	Then, by Lemma \ref{lem:PLNFisher}, there exists a function $K(\by)$ such that 
	$|\F_{(g,i,j,g,i',j')}(\bnu, \by)| \leq K( \by)$ and $\Ex(K(\by)) < \infty$. 
	The same proof  can be applied to the $g \neq g'$ case.
\end{proof}

\begin{proof}[of Lemma \ref{lem:tr}]
	By Lemma \ref{lem:poly},
	$$
	\Ex\left[\prod_{j=1}^p \phi({n_j},y_j) \tr\left(\bT\bTheta^{-1}\right)\right] = \tr\left(\bT\bTheta^{-1}\right)\exp\left({{N}}\trans \bmu + {{N}}\trans \bTheta^{-1} {{N}}/2\right). 
	$$
	Similar to the proof of Lemma S1, by  the moment generating function of the normal distribution, we have
	\begin{align*}
		&\quad \Ex\left[\prod_{j=1}^p \phi({n_j},y_j) \tr \left\lbrace \bT({\bx} -\bmu) ({\bx} -\bmu)\trans\right\rbrace\right] = \Ex_\bx\left[ \exp({N}\trans \bx)\tr \left\lbrace \bT({\bx} -\bmu) ({\bx} -\bmu)\trans \right\rbrace \right] \\ 
		&= \left\lbrace \tr\left(\bT\bTheta^{-1}\right) + {{N}}\trans \bTheta^{-1}\bT \bTheta^{-1} {{N}}\right\rbrace \exp\left({{N}}\trans \bmu + {{N}}\trans \bTheta^{-1} {{N}}/2\right).
	\end{align*} 
	Lemma \ref{lem:tr} follows from the above two equations. 
\end{proof}

\begin{proof}[of Lemma \ref{lem:mixsubspace}]
	By the proof of Lemma \ref{lem:subspace}, there exists $\balpha\not\in  \bigcup_{i=1}^n \mathcal{M}_i$.  We can assume $\balpha \in \mathcal{V}$. Otherwise, we complete the proof.
	Since $A$ is not a zero matrix, we can take $ {{\beta}} \not\in \mathcal{V} $.
	On the one hand, since ${{\beta}} \not\in \mathcal{V} $,  there are at most two  integers $k \in \mathbb{N}$ satisfying  the quadratic equation $(\balpha + k{{\beta}} )\trans A (\balpha + k{{\beta}} ) = 0$.  
	On the other hand,   for any $\mathcal{M}_i, i = 1 ,\ldots, n$, there is at most one  integer $k$ such that $\balpha + k{{\beta}} \in \mathcal{M}_i$.  Otherwise, 
	if there exist $k_1\neq k_2$ and $i$ satisfying
	$\balpha + k_1{{\beta}} \in \mathcal{M}_i , \balpha + k_2{{\beta}}\in \mathcal{M}_i ,$
	then we have $(k_2 - k_1){{\beta}} \in \mathcal{M}_i $.  It follows that $\balpha \in \mathcal{M}_i$, which is contradictory to the  fact that $\balpha \not\in \bigcup_{i=1}^n \mathcal{M}_i$.  Hence, there are at most $(n+2)$ $k$ such that $\balpha + k{{\beta}} \in  \bigcup_{i=1}^n \mathcal{M}_i  \bigcup  \mathcal{V} $. Since there are infinitely many non-negative  integers,  there exists an integer $k$ such that  $\balpha + k{{\beta}} \not \in  \bigcup_{i=1}^n \mathcal{M}_i  \bigcup  \mathcal{V} $. 
\end{proof}

\subsection{Proof of Theorem 2}
Since we only prove the positive definiteness of the Fisher information, we can only get the local strong convexity.   In order  to prove the convergence rate,   we first give the consistency of MLE.  The proof is based on the M-estimator theory.
For convenience, we introduce some notations in the M-estimator theory. We let $m_{\bnu}(\by_i)=\ell(\bnu, \by_i)$, $M_n(\bnu)=n^{-1}\sum_{i=1}^N m_{\bnu}(\by_i) = -\mathcal{L}_n(\bnu)$ and $M(\bnu)=\Ex_{\bnu^*} \left\{m_{\bnu}(\by_i)\right\}$.  Let $\mathcal{D}_{0}=\{\bnu_0 \in {\mathcal{D}} \ \mid \  \Ex (m_{\bnu_0})=\sup_{\bnu} \Ex(m_{\bnu})\}$. By Jessen's inequality and the identifiablity of the mixture Poisson log-normal model, we get that $\mathcal{D}_{0}$ only contains one element $\bnu^*$.  Then we need the following two conditions and Lemma  \ref{lem:wald}. The proof of Lemma \ref{lem:wald} can be found in \cite{van2000asymptotic}.

\begin{condition}
	$\limsup_{\bnu_n\rightarrow\bnu}m_{\bnu_n}(\by)\leq m_{\bnu}(\by)\ \mbox{ for all } \bnu  \mbox{ and a.s. }\by$. 
\end{condition}

\begin{condition}
	For all sufficiently small ball $U\subset {\mathcal{D}}$, $\by\mapsto \sup_{\bnu\in U}m_{\bnu}(\by)$ is measurable and satisfies $\Ex\left\{\sup_{\bnu\in U}m_{\bnu}(\by)\right\}<\infty.$
\end{condition}

\begin{lemma}[Wald's consistency]\label{lem:wald}
	Assume that Condition (SC1-SC2) hold for $m_{\bnu}(\by)$. Suppose that $ \hat{\bnu}_n$ is any sequence of random vectors such that  $M_n( \hat{\bnu}_n)\geq M_n(\bnu_0)-o_p(1)$  for some $\bnu_0\in \mathcal{D}_{0}$. Then for any $ \epsilon>0$, and every compact set $K\subset {\mathcal{D}}$,  as $n\rightarrow \infty $, we have
	$$\pr \left\lbrace d(\hat{\bnu}_n,\mathcal{D}_0)\geq\epsilon\wedge\hat{\bnu}_n\in K\right\rbrace \rightarrow 0,$$
	where $d(\hat{\bnu}_n,\mathcal{D}_0) = \inf_{\bnu_0\in \mathcal{D}_{0}} ||\hat{\bnu}_n - \bnu_0 ||_2$.
\end{lemma}

\begin{lemma} \label{lem:consis}
	Assume that the estimator $\hat{\bnu}_n$ minimizes  (\ref{prob:opt})  in parameter space ${\mathcal{D}}$ and $\lambda_n$ goes to zero.  Then, for any $\epsilon>0$, as $n\rightarrow \infty$,  we have:
	$$\pr(||\hat{\bnu}_n-\bnu^*||_2 \geq\epsilon)\rightarrow 0, \mbox{ as } n \rightarrow \infty. $$
	
\end{lemma}

\begin{lemma} [Uniform law of large numbers] \label{lem:ULLN}
	For all $g,i,j,g',i',j'$ we have 
	$$\pr\left\lbrace \lim_{n \rightarrow \infty} \sup_{\bnu \in {\mathcal{D}}}\left| \frac{1}{n} \sum_{k=1}^n\F_{(g,i,j)(g',i',j')}(\bnu, \by_k) - \D_{(g,i,j)(g',i',j')}(\bnu)\right|  = 0\right\rbrace  = 1,$$
	where $\D(\bnu) = \Ex_{\bnu^*}(\F(\bnu, \by)).$
	Furthermore,  we have 
	$$\pr\left\lbrace \lim_{n \rightarrow \infty} \sup_{\bnu \in {\mathcal{D}}}\left\|  \frac{1}{n} \sum_{k=1}^n\F(\bnu,  \by_k) - \D(\bnu)\right\| _2 = 0\right\rbrace  = 1.$$
	
\end{lemma}

\begin{remark} \label{cont}
	By the proof of Lemma \ref{lem:ULLN}, we know that there exists a function $F_0(\by)$ such that for any $\bnu$,
	$|\F_{(g,i,j,g',i',j')}(\nu, \by)| \leq F_0(\by)$  and $\Ex(F_0(\by)) < \infty$ .  Then, by the  dominated convergence theorem, $\D(\bnu) = \Ex_{\bnu^*}(\F(\bnu, \by))$ is continuous. 
\end{remark}

Let $\Delta_n = \hat{\bnu}_{n} - \bnu^{*}$ .  We define 
$\delta \mathcal{L}_n = \mathcal{L}_n(\bnu^{*} + \Delta_n) -\mathcal{L}_n(\bnu^{*}) - \langle \nabla\mathcal{L}_n(\bnu^{*}), {\Delta_n} \rangle,$
where $\langle\nabla\mathcal{L}_n(\bnu^{*}), {\Delta_n}\rangle = \nabla\mathcal{L}_n(\bnu^{*})\trans {\Delta_n}$.  According to Lemma \ref{lem:ULLN}, we can prove the following lemma.
\begin{lemma}\label{lem:delta}
	Under Condition 1-3, with high probability,  we have
	$\delta \mathcal{L}_n \geq (\kappa/3)||{\Delta_n}||_2^2,$
	where $\kappa =  \lambda_{min}(\bGamma^*)$.
	
\end{lemma}

\begin{proof}[of Theorem 2]
	Define
	$$\mathcal{F}(\Delta) = \mathcal{L}_n(\bnu^{*} + \Delta) -\mathcal{L}_n(\bnu^{*}) + \lambda_n \left\lbrace \mathcal{R}(\bnu^* + \Delta) - \mathcal{R}(\bnu^*) \right\rbrace .$$
	Since $\mathcal{F}(0) = 0$ and $\hat{\bnu}_n$ minimizes (\ref{prob:opt}), 
	we must have $\mathcal{F}(\Delta_n) \leq 0.$
	Further, by Lemma \ref{lem:delta}, with high probability, 	$\delta \mathcal{L}_n \geq (\kappa/3) ||\Delta_n||_2^2$. Then, combining with Cauchy's inequality and the convexity of the lasso penalty, with high probability, we get
	\begin{align*}
		\mathcal{F}(\Delta_n) 
		&= \delta \mathcal{L}_n + \langle\nabla\mathcal{L}_n(\bnu^{*}), \Delta_n\rangle+ {\lambda_n} \left\lbrace  \mathcal{R}(\bnu^* + \Delta_n) - \mathcal{R}(\bnu^*)\right\rbrace  \\
		&\geq (\kappa/3) ||\Delta_n||_2^2 - (||\nabla\mathcal{L}_n(\bnu^{*})||_{\infty}  + {2\lambda_n})||\Delta_n||_1 \\
		&\geq (\kappa/3) ||\Delta_n||_2^2 -\sqrt{Gp(p+1)/2} (||\nabla\mathcal{L}_n(\bnu^{*})||_{\infty}  + {2\lambda_n})||\Delta_n||_2 .
	\end{align*}
	When 
	$||\Delta_n||_2  >  (3/\kappa) \sqrt{Gp(p+1)/2} (||\nabla\mathcal{L}_n(\bnu^{*})||_{\infty} + {2\lambda_n}),$
	we have $\mathcal{F}(\Delta_n) > 0$. By the fact that $\mathcal{F}(\Delta_n) \leq 0$.  we obtain
	$$||\Delta_n||_2  \leq   (3/\kappa) \sqrt{Gp(p+1)/2} (||\nabla\mathcal{L}_n(\bnu^{*})||_{\infty} + {2\lambda_n}),$$
	and thus we prove this theorem.
\end{proof}

\begin{proof}[of Lemma \ref{lem:consis}]
	We use  Lemma \ref{lem:wald} to prove this lemma.  We first check Condition (SC1-SC2).  Observe that  $m_{\bnu}(\by)= \ell(\bnu, \by) = \log\left\{L_M(\bnu, \by)\right\} + C(\by) $ where $C(\by) = \sum_{j=1}^p \log (y_j!) - 2^{-1}p \log(2\pi)$ and $L_M(\bnu, \by)$ is continuous at all $  \bnu \in {\mathcal{D}}$ for any fixed $\by$. Condition (SC1) thus follows.
	
	For Condition (SC2), we first check the measurability of $\sup_{\bnu\in U}m_{\bnu}(\by)$ for any small ball $U$. 
	Let $Q_U=\{\bnu \mid \bnu\ \mbox{ is a rational point }, \bnu \in U\}$. Then $Q_U$ has a countable  number of elements.  From the measurability of $m_{\bnu}(\by)$, we get that $\sup_{\bnu\in Q_U}m_{\bnu}(\by)$ is measurable. 
	On the other hand, from the continuity of $m_{\bnu}(\by)$ in $ \bnu$, we  get
	$\sup_{\bnu\in Q_U}m_{\bnu}(\by)=\sup_{\bnu\in U}m_{\bnu}(\by),$
	and hence $\sup_{\bnu\in U}m_{\bnu}(\by)$ is measurable.
	Finally, we  prove $\Ex \left\{\sup_{\bnu\in U}m_{\bnu}(\by)\right\} <\infty$. We have
	$$\begin{array}{ll}
		&\quad \log \left\lbrace L_M(\bnu, \by)\right\rbrace  \\
		&= \log\left(\sum_{g=1}^{G}\pi_g \int {\exp\left[ \sum_{j=1}^{p}\left\lbrace  x_{j}y_{j}-\exp\left( x_{j}\right) \right\rbrace  -2^{-1}(\bx-\bmu_g)^T\bTheta_g(\bx-\bmu_g) \right]}\det(\bTheta_g)^{1/2}d\bx\right)\\
		&\leq \log\left[\sum_{g=1}^{G}\pi_g \int {\exp \left\lbrace  \bx^T\by -2^{-1}(\bx-\bmu_g)^T\bTheta_g(\bx-\bmu_g)\right\rbrace }\det(\bTheta_g)^{1/2}d\bx\right]\\
		&=\log \left\lbrace  \sum_{g=1}^{G}\pi_g \exp\left( 2^{-1}\by^T\bTheta_g^{-1}\by+\by^T\mu_g\right) \right\rbrace  + 2^{-1}p \log(2\pi) . 
	\end{array}$$
	and $\sup_{\bnu\in {\mathcal{D}}}(1/2)\by^T\bTheta_g^{-1}\by+\by^T\bmu_g\leq\|\by\|_2^2/(2M)+M\|\by\|_1$. Then, we have 
	$$\sup_{\bnu\in {\mathcal{D}}}\log\left\lbrace L_M(\bnu, \by)\right\rbrace \leq \|\by\|^2_2/(2M)+M\|\by\|_1 + 2^{-1}p \log(2\pi).$$
	Also,  we have 
	$C(\by) = \sum_{j=1}^p \log (y_j!) - 2^{-1}p \log(2\pi) \leq  \sum_{j=1}^p  y_j \log(y_j + 1)- 2^{-1}p \log(2\pi). $
	Since any polynomial of a MPLN random variable $\by$ is integrable, we prove  $\Ex(\sup_{\bnu\in U}m_{\bnu}(\by))<\infty.$
	In addition, we have
	$$\begin{array}{ll}
		M_n(\hat{\bnu}_n)&\geq \lambda_n \sum_{g=1}^{G}\|\hat{\bTheta}_g\|_{1,\rm off}+ M_n(\bnu^*)-\lambda_n \sum_{g=1}^{G}\|\bTheta_{g}^*\|_{1,\rm off}\\
		&\geq M_n(\bnu^*)-\lambda_n \sum_{g=1}^{G}\|\bTheta_{g}^*\|_{1,\rm off}=M_n(\bnu^*)-o(1).
	\end{array}$$
	Thus, all conditions in Lemma \ref{lem:wald} are satisfied. 
	Finally, observe that $\mathcal{D}_{0}$ only contains one element $\bnu^*$. Taking $K = {\mathcal{D}} $,  we get, for any $\epsilon>0$,  
	$\pr(||\hat{\bnu}_n-\bnu^*||_2 \geq\epsilon)\rightarrow 0, \mbox{ as } n \rightarrow \infty, $
	and thus we complete the proof.
\end{proof}

\begin{proof}[of Lemma \ref{lem:ULLN}]
	By Theorem  16(a) in \cite{ferguson2017course},  for any $g,i,j,g',i',j'$ 
	we only need to verify  there exists a function $F_0(\by)$ such that 
	$|\F_{(g,i,j)(g',i',j')}(\bnu, \by)| \leq F_0(\by)$ and $\Ex_{\bnu^*}(F_0(\by)) < \infty$.   The existence of such function is guaranteed by Lemma \ref{lem:MPLNFisher} .
\end{proof}

\begin{proof}[of Lemma \ref{lem:delta}]
	By the definition of $\mathcal{L}_n$ and Taylor expansion, we have 
	\begin{align*}
		\delta \mathcal{L}_n &= \mathcal{L}_n(\bnu^{*} + \Delta_n) -\mathcal{L}_n(\bnu^{*}) - \langle\nabla\mathcal{L}_n(\bnu^{*}),\Delta_n\rangle = -(1/n) \sum_{i=1}^n\Delta_n \trans \F(\check{\bnu}, \by_i)\Delta_n \\
		&=\Delta_n \trans \left\lbrace \D(\check{\bnu})- (1/n) \sum_{i=1}^n F(\check{\bnu}, \by_i)\right\rbrace {\Delta_n} + 
		{\Delta_n} \trans \left\lbrace \D(\bnu^*)-\D(\check{\bnu})\right\rbrace {\Delta_n}
		+ \Delta_n \trans \left\lbrace -\D(\bnu^*)\right\rbrace {\Delta_n}
	\end{align*} 
	where $\check{\bnu} = \bnu^{*} + \theta \Delta_n, 0 \leq\theta \leq 1.$
	By Lemma \ref{lem:ULLN}, we get, for any $\epsilon$,  there exists $N$ such that when $n > N$, with probability $1-\epsilon$, 
	\begin{equation}\label{equ:delta1}
		\Delta_n \trans \left\lbrace \D(\check{\bnu})- (1/n) \sum_{i=1}^n \F(\check{\bnu}, \by_i)\right\rbrace \Delta_n \geq -(\kappa/3)||\Delta_n||_2^2. 
	\end{equation}
	Also, by the continuity  of $\D(\bnu)$ (Remark \ref{cont}) and Wielandt-Hoffman Theorem \citep{bhatia2013matrix}, there exists a constant $\tau$ such that when $||\Delta_n||_2 < \tau$,  
	we have $$  \lambda_{min} \left\lbrace \D(\bnu^*)-\D(\check{\bnu})\right\rbrace  \geq -(\kappa/3). $$ It follows that 
	\begin{equation}\label{equ:delta2}
		\Delta_n \trans \left\lbrace \D(\bnu^*)-\D(\check{\bnu})\right\rbrace \Delta_n \geq -(\kappa/3) ||\Delta_n||_2^2. 
	\end{equation}
	Finally, by $\kappa = \lambda_{min}(\bGamma^*)$, we have 
	\begin{equation}\label{equ:delta3}
		\Delta_n \trans \left\lbrace -\D(\bnu^*)\right\rbrace \Delta_n \geq \kappa ||\Delta_n||_2^2.
	\end{equation}
	Combining the above inequalities (\ref{equ:delta1}-\ref{equ:delta3}), we have  when  $||\Delta_n||_2 < \tau$, with high probability, 
	$\delta \mathcal{L}_n \geq (\kappa/3) ||\Delta_n||_2^2.$
	Finally, since $\tau$ is a constant, by the consistency of the mixture Poisson log-normal model, we have 
	$\pr(||{\Delta}_n|||_2 < \tau) \rightarrow 1 \mbox{ as } n \rightarrow \infty$. Then, we get 
	$\delta \mathcal{L}_n \geq (\kappa/3) ||{\Delta}_n|||_2^2.$
	Thus we complete the proof.
\end{proof}

\subsection{Proof of Theorem 3}
In this subsection,  we simplify the notation and use $S$ and $S^c$ to denote $S(\bnu^*)$ and $S^c(\bnu^*)$, respectively. Define
$\hat{\bGamma}^*_i = -\F(\bnu^*, \by_i) $
and $\hat{\bGamma}^* = (1/n) \sum_{i=1}^n \hat{\bGamma}^*_i$, which is an estimator of $\bGamma^*$. 
We write $\alpha = 1-||\bGamma^*_{S^cS}{(\bGamma^*_{SS})}^{-1}||_{1,\infty}$. 	By Condition 4, $\alpha > 0$.
\begin{lemma} \label{samplegamma}
	For any $\epsilon > 0$, with high probability, $\hat{\bGamma}_{SS}^*$ is invertible and
	$||\hat{\bGamma}^*_{S^cS}{(\hat{\bGamma}^*_{SS})}^{-1}||_{1,\infty} \leq 1-\alpha/2.$
\end{lemma}

\begin{lemma}\label{lem:KKT}
	For any $\lambda_n > 0$, with high probability, the solution $\hat{\nu}_n$ to the optimization problem (\ref{prob:opt}) is characterized by
	$$-n^{-1}\sum_{i=1}^n\mathcal{S}^M(\hat{\bnu}_n, \by_i)  + \lambda_n {\hat{\bK}} = 0,$$
	where ${\hat{\bK}}$ is the subdifferential of $\mathcal{R}(\bnu)$ at $\hat{\bnu}_n$.
\end{lemma}

Next, we construct the primal-dual witness solution $(\widetilde{\bnu}_n, {\widetilde{\bK}} )$. Let $\widetilde{\bnu}_n$ be  the solution to the restricted  optimization problem
$$\widetilde{\bnu}_n \in \arg \min_{\bnu_{S^c} = 0,  \bnu \in {\mathcal{D}}} \mathcal{L}_n(\bnu) + {\lambda_n} \mathcal{R}(\bnu).$$
Define 
\begin{equation} \label{eq:tildeK}
	\widetilde{\bK} = \lambda_n^{-1} \left\lbrace n^{-1}\sum_{i=1}^n\mathcal{S}^M(\widetilde{\bnu}_n, \by_i)  \right\rbrace. 
\end{equation} 
Then, we have
\begin{equation} \label{equ:res}
	-n^{-1}\sum_{i=1}^n\mathcal{S}^M(\widetilde{\bnu}_n, \by_i)  + \lambda_n {\widetilde{\bK}} = 0.
\end{equation}
Define
$\widetilde{\Delta} _n= \bnu^* - \widetilde{\bnu}_n$. 
We rewrite (\ref{equ:res}) as  
$$-n^{-1}\sum_{i=1}^n\mathcal{S}^M(\widetilde{\bnu}_n, \by_i)  + n^{-1}\sum_{i=1}^n\mathcal{S}^M(\bnu^*, \by_i)-n^{-1}\sum_{i=1}^n\mathcal{S}^M(\bnu^*, \by_i) + \hat{\bGamma}^*\widetilde{\Delta} _n- \hat{\bGamma}^*\widetilde{\Delta} _n+ \lambda_n {\widetilde{\bK}}  = 0.$$
Let 
$${\R} = -n^{-1}\sum_{i=1}^n\mathcal{S}^M(\widetilde{\bnu}_n, \by_i)  + n^{-1}\sum_{i=1}^n\mathcal{S}^M(\bnu^*, \by_i) + \hat{\bGamma}^*\widetilde{\Delta}_n,$$
and ${\bW} = -n^{-1}\sum_{i=1}^n\mathcal{S}^M(\bnu^*, \by_i)$.
Therefore, we have 
\begin{equation} \label{KKT}
	{\R} - \hat{\bGamma}^*\widetilde{\Delta} _n+ {\bW} + \lambda_n {\widetilde{\bK}} = 0.
\end{equation}
Since we  restrict  the solution to the set of true support, similarly to Theorem 1, we can proceed analogously to the proof of 
\begin{equation} \label{eq:dual}
	\pr\left\{||\widetilde{\nu}_n-\nu^*||_2 \leq  (3/\kappa) \sqrt{G{p(p+1)}/2}\left(n^{-1} ||\nabla{\ell}_n(\bnu^*)||_{\infty} + 2{\lambda_n}{}\right)\right\}\rightarrow 1, \mbox{ as } n \rightarrow \infty.
\end{equation}
Considering the local convex property of loss function at $\bnu^*$ in Lemma \ref{lem:delta}, if we verify that with high probability, the strict dual feasibility condition
$||\widetilde{\bK}||_{\infty} \leq 2$ holds,  then we prove that with high probability $\widetilde{\bnu}_n $ is equal to $\hat{\bnu}_n$. Then, the model can recover all zeros. Similarly to Lemma \ref{lem:KKT}, by the fact that $\widetilde{\bnu}_n$ is the restricted construction and the definition of $\widetilde{\bK}$ in (\ref{eq:tildeK}), we have 
\begin{equation}\label{eq:tildeKS}
	||\widetilde{\bK}_{S}||_{\infty} \leq 2 .
\end{equation} 
Therefore, we only need to show that $||\widetilde{\bK}_{S^c}||_{\infty} \leq 2 $.  
Observe that the infinity norm of $\widetilde{\bK}_{S^c}$ is less than 2 instead of 1. The reason is  $\mathcal{R}(\bnu) =  \sum_{g=1}^G ||\bTheta_g||_{1,\rm off} = 2\sum_{g=1}^G \sum_{i<j} |\bTheta_{g,ij}|$. 
Hence, we aim to  verify  the strict dual feasibility. 

\begin{lemma}[Strict dual feasibility] \label{dualf}
	Under Condition 1-4,  suppose 
	that $\hat{\bGamma}_{SS}^*$ is invertible and
	$$||{\bW}||_{\infty} + ||{\R}||_{\infty} < \alpha \lambda_n/4,\,\quad \left\| \hat{\bGamma}^*_{S^cS}{\left(\hat{\bGamma}^*_{SS}\right)}^{-1}\right\|_{1,\infty} \leq 1-\alpha/2.$$
	Then, the matrix $\widetilde{\bK}$ satisfies 
	$||\widetilde{\bK}_{S^c}||_{\infty} \leq 2.$
\end{lemma}

\begin{lemma} \label{lem:remainder}
	$||{\R}||_{\infty}/||\widetilde{\Delta}_n||_2\rightarrow 0$, in probability as $n \rightarrow  \infty$.
\end{lemma}

Applying  Chebyshev's inequality, the following lemma is clear. 
\begin{lemma} \label{noiseterm}
	Let $a_n$ be any sequence such that ${a_n}{\sqrt{n}} \rightarrow \infty, a_n > 0$ and $a_n \rightarrow 0$. Then, we have 
	$$\pr(||{\bW}||_{\infty} \leq a_n) \rightarrow 1, \mbox{ as } n \rightarrow \infty.$$
\end{lemma}

\begin{proof}[of Theorem 3]
	Observe that ${\bW} =-n^{-1} \nabla \ell_n(\bnu^{*})$ and $\widetilde{\Delta} _n= \bnu^* - \widetilde{\bnu}_n$. 
	Applying Lemma \ref{lem:remainder}, we have 
	$  ||{\R}||_{\infty} = o_p(1) ||\widetilde{\Delta}_n||_2. $ 
	By (\ref{eq:dual}), we have 
	$$||{\R}||_{\infty} = o_p(1) \left\lbrace  (3/\kappa) \sqrt{G{p(p+1)}/2} \left(||{\bW}||_{\infty} +2 {\lambda_n}\right)\right\rbrace  .$$
	It follows that 
	\begin{align*}
		||{\bW}||_{\infty} + ||{\R}||_{\infty} & =||{\bW}||_{\infty} +  o_p(1) \left\lbrace  (3/\kappa) \sqrt{G{p(p+1)}/2} \left(||{\bW}||_{\infty} +2 {\lambda_n}\right)\right\rbrace \\
		& = \left\lbrace 1+o_p(1)\right\rbrace ||{\bW}||_{\infty} + o_p(1)\lambda_n.
	\end{align*}
	In order to show   $||{\bW}||_{\infty} + ||{\R}||_{\infty} < \alpha \lambda_n/4,$ we only need to show that with high probability, 
	$$||{\bW}||_{\infty} \leq \left\lbrace 1+o_p(1)\right\rbrace ^{-1} \left\lbrace \alpha/4-o_p(1)\right\rbrace \lambda_n.$$
	By the  choice of $\lambda_n$ such that $\lambda_n \rightarrow 0$ and ${\sqrt{n}\lambda_n} \rightarrow \infty$,
	applying Lemma \ref{noiseterm}, we have 
	$$\pr \left[ ||{\bW}||_{\infty} \leq \left\lbrace 1+o_p(1)\right\rbrace ^{-1}\left\lbrace \alpha/4-o_p(1)\right\rbrace \lambda_n\right]  \rightarrow 1.$$ 
	By Lemma  \ref{samplegamma}, we have with high probability 	$\hat{\bGamma}_{SS}^*$ is invertible and 
	$ \left\|\hat{\bGamma}^*_{S^cS}{\left(\hat{\bGamma}^*_{SS}\right)}^{-1} \right\|_{1,\infty} \leq 1-\alpha/2.$
	Then, by Lemma \ref{dualf}, with high probability, we have  the strict dual feasibility condition holds. Hence, the witness solution $\widetilde{\bnu}_n$ is equal to the original solution $\hat{\bnu}_n$. 
	Since $\widetilde{\bnu}_n$ is the restricted solution,  with high probability, $\widetilde{\bnu}_n$ can recover all zeros. It follows that $\hat{\bnu}_n$ can recover all zeros.  Finally, since $ ||\hat{\bnu}_n-\bnu^*||_2  \rightarrow 0 $ in probability, with high probability, $\hat{\bnu}_n$ can recover all non zeros.
\end{proof}

\begin{proof}[of Lemma \ref{lem:KKT}]
	Since $\bnu^*$  is an interior point of $\mathcal{D}$ ,  we only need to prove with high probability, the maximum value will not be taken at $\partial{\mathcal{D}}$.  By Jessen's inequality and the identifiablity of the mixture Poisson log-normal model, we have  
	$\sup_{\bnu \in \partial{\mathcal{D}}} 
	\Ex_{\bnu^*} \left[\log \left\{ p\left(\by_i;\bnu, \{\bmu_g\}_{g=1}^G \right) / p\left(\by_i;\bnu^*, \{\bmu_g\}_{g=1}^G\right) \right\} \right] < 0$. By uniform law of large numbers, we have  with high probability,
	$$\sup_{\bnu \in \partial{\mathcal{D}}}\left\lbrace n^{-1}\sum_{i=1}^n \log \, p\left(\by_i;\bnu, \{\bmu_g\}_{g=1}^G\right) - n^{-1}\sum_{i=1}^n \log \, p\left(\by_i;\bnu^*, \{\bmu_g\}_{g=1}^G\right)\right\rbrace  < 0.$$
	Thus we complete the proof.  
\end{proof}

\begin{proof}[of Lemma \ref{dualf}]
	We split $\widetilde{\Delta}_n$ into $\widetilde{\Delta}_{nS} $ and $\widetilde{\Delta}_{nS^c} $. 
	By $\widetilde{\Delta}_{nS^c} = 0$,  (\ref{KKT})  can be rewritten as two blocks of linear equations 
	\begin{equation}\label{part2}
		{\R}_S - \hat{\bGamma}_{SS}^*\widetilde{\Delta}_{nS} + {\bW}_S + \lambda_n {\widetilde{\bK}}_{S} = 0,
	\end{equation}
	\begin{equation}\label{part}
		{\R}_{S^c} - \hat{\bGamma}_{S^cS}^*\widetilde{\Delta}_{nS} + {\bW}_{S^c} + \lambda_n {\widetilde{\bK}}_{S^c} = 0.
	\end{equation}
	From (\ref{part2}), we have
	$\widetilde{\Delta}_{nS} ={\left(\hat{\bGamma}^*_{SS}\right)}^{-1}\left({\R}_S + {\bW}_S + \lambda_n {\widetilde{\bK}}_S\right).$
	Substituting this expression into (\ref{part}), with high probability, we have
	$${\widetilde{\bK}}_{S^c} = \lambda_n^{-1}	\left\lbrace -{\R}_{S^c} + \hat{\bGamma}_{S^cS}^*{\left(\hat{\bGamma}^*_{SS}\right)}^{-1}\left({\R}_S + {\bW}_S + \lambda_n {\widetilde{\bK}}_S\right) - {\bW}_{S^c} \right\rbrace .$$
	Let $\bA$ be a matrix and ${a}$ be a vector.
	By the fact that $||\bA {a}||_{\infty} \leq ||\bA||_{1,\infty} ||{a}||_{\infty}$, we have
	\begin{align*}
		||{\widetilde{\bK}}_{S^c}||_{\infty} &\leq {\lambda_n}^{-1}\left(1 + \bigg|\bigg|\hat{\bGamma}_{S^cS}^* {\left(\hat{\bGamma}^*_{SS}\right)}^{-1}\bigg|\bigg|_{1,\infty}\right)\left(||{\bW}||_{\infty} + ||{\R}||_{\infty}\right)  + \bigg|\bigg|\hat{\bGamma}_{S^cS}^* {\left(\hat{\bGamma}^*_{SS}\right)}^{-1}\bigg|\bigg|_{1,\infty} ||{\widetilde{\bK}}_{S}||_{\infty} \\
		&\leq  2  (1-\alpha/2) + ((2-\alpha/2)/\lambda_n)(||{\bW}||_{\infty} + ||{\R}||_{\infty}) \leq 2.
	\end{align*}
	Here we use the inequality that $||{\widetilde{\bK}}_{S}||_{\infty} \leq 2$ in (\ref{eq:tildeKS}). Thus we complete the proof.
\end{proof}

\begin{proof}[of Lemma \ref{lem:remainder}]
	By the mean value theorem, we have
	$${\R} = -n^{-1}\sum_{i=1}^n\mathcal{S}^M(\widetilde{\bnu}_n, \by_i)  + n^{-1}\sum_{i=1}^n\mathcal{S}^M(\bnu^*, \by_i) + \hat{\bGamma}^*\widetilde{\Delta} _n = \left\lbrace \hat{\bGamma}^* + n^{-1} \sum_{i=1}^n \F(\check{\bnu}, \by_i) \right\rbrace\widetilde{\Delta} _n,$$
	where $\check{\bnu} = \bnu^{*} + \theta (-\widetilde{\Delta} _n), 0 \leq\theta \leq 1.$
	Since $||\R||_{\infty} /||\widetilde{\Delta} _n||_2 \leq ||{\R}||_2 / ||\widetilde{\Delta} _n||_2$,
	we have 
	$$||\R||_{\infty} /||\widetilde{\Delta} _n||_2 \leq \left\|\hat{\bGamma}^*+ n^{-1} \sum_{i=1}^n \F(\check{\bnu}, \by_i)  \right\|_2.$$
	By the triangle inequality, we have	
	\begin{align*}
		&\quad \bigg|\bigg|\hat{\bGamma}^*+ n^{-1} \sum_{i=1}^n  \F(\check{\bnu}, \by_i) \bigg|\bigg|_2 = \bigg|\bigg|\hat{\bGamma}^*+\D(\bnu^*)-\D(\bnu^*) +\D(\check{\bnu})-\D(\check{\bnu}) + n^{-1} \sum_{i=1}^n \F(\check{\bnu}, \by_i) \bigg|\bigg|_2 \\
		&\leq ||\hat{\bGamma}^*+\D(\bnu^*)||_2+||-\D(\bnu^*) +\D(\check{\bnu})||_2+\bigg|\bigg|-\D(\check{\bnu}) + n^{-1} \sum_{i=1}^n  \F(\check{\bnu}, \by_i) \bigg|\bigg|_2.
	\end{align*}
	By $||\bnu^*-\check{\bnu} ||_2^2 \rightarrow 0$ in probability and Lemma \ref{lem:ULLN}, we have $||\R||_{\infty} /||\widetilde{\Delta} _n||_2\rightarrow 0$ in probability.
\end{proof}

\section{Appendix: Simulation}

\subsection{Details of the Data Generation Process}
\label{Simulation:datageneration}
For each simulation dataset, we first independently generate the precision matrix for each of the 3 latent normal distributions according to one of the four graph structures. When generating the precision matrices, the diagonal elements are set as 1 plus a small positive number to guarantee positive definiteness. Then, we generate the mean vectors $\bmu_1,\bmu_2,\bmu_3$ for the latent normal distributions. The first $p_d$ elements of $\bmu_g$ ($g=1,2,3$) are independently sampled from $\{{{v}_{1}},\left({{v}_{1}}+{{v}_{2}}\right)/2,{{v}_{2}}\}$. The remaining $p-p_d$ elements are shared among $\bmu_1,\bmu_2, \bmu_3$ and are independently sampled from $\{{v}_{3},{v}_{4}\}$. We set $\left(v_1,v_2,v_3,v_4\right)$ as  $(2.4,-0.1,0.9,-0.1)$  in the low dropout case (about $10\%$ zeros) and $(1.4,-1.1,-0.1,-1.1)$ in the high dropout case  (about $40\%$ zeros). We vary $p_d$ to control the mixing degree of the three populations. The scaling factors $\bl = (l_1,\dots,l_n)$ are independently generated from a log-normal distribution $\log \mbox{N}(\log 10,0.05)$. With these model parameters, we finally generate the observed expression $\bY_1,\dots,\bY_n$ from the mixture Poisson log-normal model. We calculate the Adjusted Rand Index between the true population label and  population label from the K-means clustering \citep{hartigan1979algorithm} of the normalized data $\tilde{\bY} = \log (\bY+1) - (\log \hat{\bl}) 1_{p}^{T}$ with $\hat{l}_i = \sum_{j=1}^{p} Y_{i j} / 10^4$ $(i = 1,2,\dots,n)$. We vary $p_d$ such that the low-level mixing data have an Adjusted Rand Index value in  $\left(0.9,1\right]$, the middle-level mixing data have an Adjusted Rand Index value in $\left(0.75,0.85\right]$ and the high-level mixing data have an Adjusted Rand Index value in $\left(0.65,0.75\right]$.

\subsection{The Description of Edge Scores, Partial Precision-Recall Curve, and Parameter Selection for different algorithms}
\label{Simulation:furtherint}

\textbf{1. Edge scores:} For VMPLN, VPLN and Glasso, suppose that $\hat{\Theta}$ is its estimation of a network, we define a edge score for the edge $(i,j)$ ($i\neq j$) as its absolute partial correlation, i.e. $\big|-(\hat{\Theta}_{ii} \hat{\Theta}_{jj})^{-\frac{1}{2}} \hat{\Theta}_{ij}\big|$. For LPGM, we define a edge score for each edge as its stability score. For PPCOR, GENIE3 and PIDC, we define a edge score for each edge as its estimated connected weight.\\
\textbf{2. Partial precision-recall curve:} Since the network inferred by the available method contain connected edges with different edge scores and unconnected edges with zero edge scores, the precision-recall curve constructed by varying threshold of the selected edges are incomplete and we can only obtain a partial precision-recall curve. We calculate its area under this partial precision-recall curve (pAUPRC). In order to eliminate the influence of different network densities given by the avilable methods, we further define the pAUPRC ratio as the ratio between the pAUPRC and the expected pAUPRC of the random network prediction with the same network density for a fair comparison.\\
\textbf{3. Parameter selection for different algorithms:} 

Selecting parameters using each algorithm's default: We use default parameters for PPCOR, GENIE3 and PIDC. We tune the parameter of VMPLN using the integrated complete likelihood criterion, the parameters of VPLN and Glasso using the Bayesian information criterion, and the parameters of LPGM using the stability.

Selecting parameters such that the density of the estimated network is 20\%: VMPLN, VPLN, Glasso and
LPGM, we first tune their tuning parameters such that the densities of the estimated networks are
20\%. For PPCOR, GENIE3 and PIDC, we select the edge score cutoffs such that the estimated network densities are 20\%.

\section{Appendix: Real Data Analysis}
\subsection{The used public gene regulatory databases}
\label{Realdata:GRNdb}
The used public gene regulatory databases include \\
\textbf{PPI databases}: STRING \citep{szklarczyk2019string}, HumanTFDB \citep{hu2019animaltfdb}\\
\textbf{ChIP-seq databases}: hTFtarget \citep{zhang2020htftarget}, ChEA \citep{lachmann2010chea}, ChIP-Atlas\citep{oki2018ch}, ChIPBase \citep{zhou2016chipbase}, ESCAPE \citep{xu2013escape}\\
\textbf{Integrated databases}: TRRUST \citep{han2018trrust}, RegNetwork \citep{liu2015regnetwork}.

\subsection{Silver Standard Construction for Benchmarking on scRNA-seq data}
\label{Realdata:silverstandard}

The Kang data consists of two batches, the interferon Beta 1 (IFNB1)-stimulated and control groups. The Zheng data also consists of two batches, which are respectively sequenced by $3^\prime$ and $5^\prime$ scRNA-seq technologies. Silver standards are constructed using the IFNB1-stimulated group and the $3^\prime$ batch for the Kang data and Zheng data, respectively. The gene pairs that occur in the public gene regulatory network databases (See appendix Section \ref{Realdata:GRNdb}) are taken as potential regulatory relationships. Each of these  potential regulatory relationships involves at least one transcription factor. Then, for each cell type in the construction batch, we calculate the Spearman's $\rho$ correlation between the gene pairs having potential regulatory relationships. If a gene pair has a significant Spearman's $\rho$ correlation, we consider the gene pair having a true regulatory relationship and add the edge to the silver standard edge set of the cell type. 

\subsection{Additional details of gene regulatory network inference of SARS-COV-2 dataset}
\label{Realdata:SARSCOV2}
We select top 2000 highly variable genes for each patient and use the union of the highly variable genes from all patients for VMPLN analysis. The gene set of interest for gene regulatory network inference is selected as the overall top 1000 highly variable genes as defined by Seurat \citep{stuart2019comprehensive}. The edges that do not appear in the public gene regulatory network databases listed in appendix Section \ref{Realdata:GRNdb} are set to 0. We only focus on the gene regulatory network among genes in the gene set of interest for gene regulatory network inference. We perform VMPLN analysis for each patient separately and select the parameters such that the density of the estimated networks (i.e. the number of inferred edges divided by the number of edges in the prior gene regulatory network set) was ~5\%. Then, for each macrophage group, we weighted average the estimated partial correlations from each moderate patients with the number of cells as weight to get the gene regulatory network under the moderate condition. Similarly, we obtain the gene regulatory network for every macrophage group under the severe condition.



\begin{figure}[htbp]
	\centering
	\includegraphics[width = 0.9\textwidth, height=0.88\textwidth]{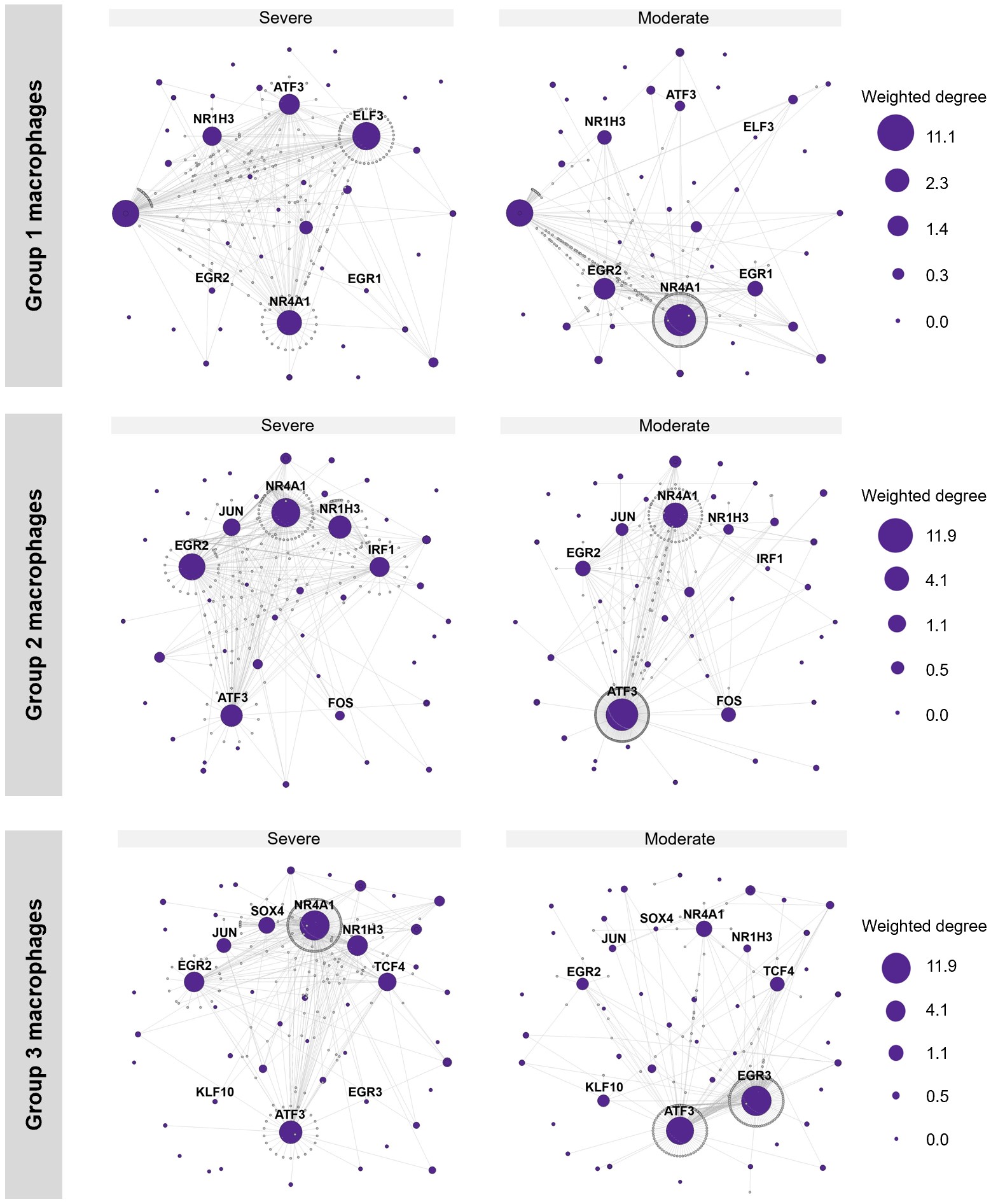}
	\caption{The inferred gene regulatory networks of Group1, Group2 and Group3 macrophages in severe and moderate patients.}
	\label{GRN-group123}
\end{figure}


\vskip 0.2in
\bibliography{VMPLN}

\begin{thebibliography}{53}
\providecommand{\natexlab}[1]{#1}
\providecommand{\url}[1]{\texttt{#1}}
\expandafter\ifx\csname urlstyle\endcsname\relax
  \providecommand{\doi}[1]{doi: #1}\else
  \providecommand{\doi}{doi: \begingroup \urlstyle{rm}\Url}\fi

\bibitem[Aibar et~al.(2017)Aibar, Gonz{\'a}lez-Blas, Moerman,
  et~al.]{aibar2017scenic}
Sara Aibar, Carmen~Bravo Gonz{\'a}lez-Blas, Thomas Moerman, et~al.
\newblock {SCENIC: single-cell regulatory network inference and clustering}.
\newblock \emph{Nature Methods}, 14\penalty0 (11):\penalty0 1083--1086, 2017.

\bibitem[Allen and Liu(2013)]{allen2013local}
Genevera~I Allen and Zhandong Liu.
\newblock {A local Poisson graphical model for inferring networks from
  sequencing data}.
\newblock \emph{IEEE Transactions on NanoBioscience}, 12\penalty0 (3):\penalty0
  189--198, 2013.

\bibitem[Barab{\'a}si and Albert(1999)]{barabasi1999emergence}
Albert-L{\'a}szl{\'o} Barab{\'a}si and R{\'e}ka Albert.
\newblock {Emergence of scaling in random networks}.
\newblock \emph{Science}, 286\penalty0 (5439):\penalty0 509--512, 1999.

\bibitem[Bhatia(2013)]{bhatia2013matrix}
Rajendra Bhatia.
\newblock \emph{{Matrix analysis}}, volume 169.
\newblock Springer Science \& Business Media, 2013.

\bibitem[Biernacki et~al.(2000)Biernacki, Celeux, and
  Govaert]{biernacki2000assessing}
Christophe Biernacki, Gilles Celeux, and G{\'e}rard Govaert.
\newblock {Assessing a mixture model for clustering with the integrated
  completed likelihood}.
\newblock \emph{IEEE Transactions on Pattern Analysis and Machine
  Intelligence}, 22\penalty0 (7):\penalty0 719--725, 2000.

\bibitem[Boyd et~al.(2011)Boyd, Parikh, Chu, et~al.]{boyd2011distributed}
Stephen Boyd, Neal Parikh, Eric Chu, et~al.
\newblock {Distributed optimization and statistical learning via the
  alternating direction method of multipliers}.
\newblock \emph{Foundations and Trends in Machine Learning}, 3\penalty0
  (1):\penalty0 1--122, 2011.

\bibitem[Cai et~al.(2011)Cai, Liu, and Luo]{cai2011constrained}
Tony Cai, Weidong Liu, and Xi~Luo.
\newblock A constrained $\ell_1$ minimization approach to sparse precision
  matrix estimation.
\newblock \emph{Journal of the American Statistical Association}, 106\penalty0
  (494):\penalty0 594--607, 2011.

\bibitem[Chan et~al.(2006)Chan, Siu, Chin, et~al.]{chan2006modulation}
Ching-Ping Chan, Kam-Leung Siu, King-Tung Chin, et~al.
\newblock {Modulation of the unfolded protein response by the severe acute
  respiratory syndrome coronavirus spike protein}.
\newblock \emph{Journal of Virology}, 80\penalty0 (18):\penalty0 9279--9287,
  2006.

\bibitem[Chan et~al.(2017)Chan, Stumpf, and Babtie]{chan2017gene}
Thalia~E Chan, Michael~PH Stumpf, and Ann~C Babtie.
\newblock {Gene regulatory network inference from single-cell data using
  multivariate information measures}.
\newblock \emph{Cell Systems}, 5\penalty0 (3):\penalty0 251--267, 2017.

\bibitem[Chiquet et~al.(2019)Chiquet, Robin, and
  Mariadassou]{chiquet2019variational}
Julien Chiquet, Stephane Robin, and Mahendra Mariadassou.
\newblock {Variational inference for sparse network reconstruction from count
  data}.
\newblock In \emph{Proceedings of the 36th International Conference on Machine
  Learning}, pages 1162--1171, 2019.

\bibitem[Dobra and Lenkoski(2011)]{dobra2011copula}
Adrian Dobra and Alex Lenkoski.
\newblock {Copula Gaussian graphical models and their application to modeling
  functional disability data}.
\newblock \emph{The Annals of Applied Statistics}, 5\penalty0 (2A):\penalty0
  969--993, 2011.

\bibitem[Drton and Maathuis(2017)]{drton2017structure}
Mathias Drton and Marloes~H Maathuis.
\newblock Structure learning in graphical modeling.
\newblock \emph{Annual Review of Statistics and Its Application}, 4:\penalty0
  365--393, 2017.

\bibitem[Farasat et~al.(2015)Farasat, Nikolaev, Srihari, and
  Blair]{farasat2015probabilistic}
Alireza Farasat, Alexander Nikolaev, Sargur~N Srihari, and Rachael~Hageman
  Blair.
\newblock Probabilistic graphical models in modern social network analysis.
\newblock \emph{Social Network Analysis and Mining}, 5\penalty0 (1):\penalty0
  1--18, 2015.

\bibitem[Ferguson(2017)]{ferguson2017course}
Thomas~S Ferguson.
\newblock \emph{{A course in large sample theory}}.
\newblock Routledge, 2017.

\bibitem[Friedman et~al.(2008)Friedman, Hastie, and
  Tibshirani]{friedman2008sparse}
Jerome Friedman, Trevor Hastie, and Robert Tibshirani.
\newblock {Sparse inverse covariance estimation with the graphical lasso}.
\newblock \emph{Biostatistics}, 9\penalty0 (3):\penalty0 432--441, 2008.

\bibitem[Grant et~al.(2021)Grant, Morales-Nebreda, Markov,
  et~al.]{grant2021circuits}
Rogan~A Grant, Luisa Morales-Nebreda, Nikolay~S Markov, et~al.
\newblock {Circuits between infected macrophages and T cells in SARS-CoV-2
  pneumonia}.
\newblock \emph{Nature}, 590\penalty0 (7847):\penalty0 635--641, 2021.

\bibitem[Hafemeister and Satija(2019)]{hafemeister2019normalization}
Christoph Hafemeister and Rahul Satija.
\newblock {Normalization and variance stabilization of single-cell RNA-seq data
  using regularized negative binomial regression}.
\newblock \emph{Genome Biology}, 20\penalty0 (1):\penalty0 1--15, 2019.

\bibitem[Han et~al.(2018)Han, Cho, Lee, Yun, Kim, Bae, Yang, Kim, Lee, Kim,
  et~al.]{han2018trrust}
Heonjong Han, Jae-Won Cho, Sangyoung Lee, Ayoung Yun, Hyojin Kim, Dasom Bae,
  Sunmo Yang, Chan~Yeong Kim, Muyoung Lee, Eunbeen Kim, et~al.
\newblock {TRRUST v2: an expanded reference database of human and mouse
  transcriptional regulatory interactions}.
\newblock \emph{Nucleic Acids Research}, 46\penalty0 (D1):\penalty0 D380--D386,
  2018.

\bibitem[Hartigan and Wong(1979)]{hartigan1979algorithm}
John~A Hartigan and Manchek~A Wong.
\newblock {Algorithm AS 136: A k-means clustering algorithm}.
\newblock \emph{Journal of the Royal Statistical Society Series C (Applied
  Statistics)}, 28\penalty0 (1):\penalty0 100--108, 1979.

\bibitem[Hu et~al.(2019)Hu, Miao, Jia, Yu, Zhang, and Guo]{hu2019animaltfdb}
Hui Hu, Ya-Ru Miao, Long-Hao Jia, Qing-Yang Yu, Qiong Zhang, and An-Yuan Guo.
\newblock {AnimalTFDB 3.0: a comprehensive resource for annotation and
  prediction of animal transcription factors}.
\newblock \emph{Nucleic Acids Research}, 47\penalty0 (D1):\penalty0 D33--D38,
  2019.

\bibitem[Huynh-Thu et~al.(2010)Huynh-Thu, Irrthum, Wehenkel, and
  Geurts]{huynh2010inferring}
V{\^a}n~Anh Huynh-Thu, Alexandre Irrthum, Louis Wehenkel, and Pierre Geurts.
\newblock {Inferring regulatory networks from expression data using tree-based
  methods}.
\newblock \emph{PLoS ONE}, 5\penalty0 (9):\penalty0 e12776, 2010.

\bibitem[Janssens et~al.(2014)Janssens, Pulendran, and
  Lambrecht]{janssens2014emerging}
Sophie Janssens, Bali Pulendran, and Bart~N Lambrecht.
\newblock {Emerging functions of the unfolded protein response in immunity}.
\newblock \emph{Nature Immunology}, 15\penalty0 (10):\penalty0 910--919, 2014.

\bibitem[Jordan et~al.(1999)Jordan, Ghahramani, Jaakkola, and
  Saul]{jordan1999introduction}
Michael~I Jordan, Zoubin Ghahramani, Tommi~S Jaakkola, and Lawrence~K Saul.
\newblock {An introduction to variational methods for graphical models}.
\newblock \emph{Machine Learning}, 37\penalty0 (2):\penalty0 183--233, 1999.

\bibitem[Kang et~al.(2018)Kang, Subramaniam, Targ, et~al.]{kang2018multiplexed}
Hyun~Min Kang, Meena Subramaniam, Sasha Targ, et~al.
\newblock {Multiplexed droplet single-cell RNA-sequencing using natural genetic
  variation}.
\newblock \emph{Nature Biotechnology}, 36\penalty0 (1):\penalty0 89--94, DOI:
  https://doi.org/10.1038/nbt.4042, 2018.

\bibitem[Kim(2015)]{kim2015ppcor}
Seongho Kim.
\newblock {ppcor: an R package for a fast calculation to semi-partial
  correlation coefficients}.
\newblock \emph{Communications for Statistical Applications and Methods},
  22\penalty0 (6):\penalty0 665, 2015.

\bibitem[Lachmann et~al.(2010)Lachmann, Xu, Krishnan, Berger, Mazloom, and
  Ma'ayan]{lachmann2010chea}
Alexander Lachmann, Huilei Xu, Jayanth Krishnan, Seth~I Berger, Amin~R Mazloom,
  and Avi Ma'ayan.
\newblock {ChEA: transcription factor regulation inferred from integrating
  genome-wide ChIP-X experiments}.
\newblock \emph{Bioinformatics}, 26\penalty0 (19):\penalty0 2438--2444, 2010.

\bibitem[Li et~al.(2020)Li, Cai, and Li]{li2020transfer}
Sai Li, T~Tony Cai, and Hongzhe Li.
\newblock {Transfer learning for high-dimensional linear regression:
  Prediction, estimation, and minimax optimality}.
\newblock \emph{arXiv preprint arXiv:2006.10593}, 2020.

\bibitem[Liao et~al.(2020)Liao, Liu, Yuan, et~al.]{liao2020single}
Mingfeng Liao, Yang Liu, Jing Yuan, et~al.
\newblock {Single-cell landscape of bronchoalveolar immune cells in patients
  with COVID-19}.
\newblock \emph{Nature Medicine}, 26\penalty0 (6):\penalty0 842--844, DOI:
  https://doi.org/10.1038/s41591--020--0901--9, 2020.

\bibitem[Liu et~al.(2015)Liu, Wu, Miao, and Wu]{liu2015regnetwork}
Zhi-Ping Liu, Canglin Wu, Hongyu Miao, and Hulin Wu.
\newblock {RegNetwork: an integrated database of transcriptional and
  post-transcriptional regulatory networks in human and mouse}.
\newblock \emph{Database}, 2015, 2015.

\bibitem[Lun et~al.(2016)Lun, Bach, and Marioni]{lun2016pooling}
Aaron~TL Lun, Karsten Bach, and John~C Marioni.
\newblock {Pooling across cells to normalize single-cell RNA sequencing data
  with many zero counts}.
\newblock \emph{Genome Biology}, 17\penalty0 (1):\penalty0 1--14, 2016.

\bibitem[Meinshausen and B{\"u}hlmann(2006)]{meinshausen2006high}
Nicolai Meinshausen and Peter B{\"u}hlmann.
\newblock {High-dimensional graphs and variable selection with the lasso}.
\newblock \emph{The Annals of Statistics}, 34\penalty0 (3):\penalty0
  1436--1462, 2006.

\bibitem[Oki et~al.(2018)Oki, Ohta, Shioi, Hatanaka, Ogasawara, Okuda, Kawaji,
  Nakaki, Sese, and Meno]{oki2018ch}
Shinya Oki, Tazro Ohta, Go~Shioi, Hideki Hatanaka, Osamu Ogasawara, Yoshihiro
  Okuda, Hideya Kawaji, Ryo Nakaki, Jun Sese, and Chikara Meno.
\newblock {ChIP-Atlas: a data-mining suite powered by full integration of
  public Ch IP-seq data}.
\newblock \emph{EMBO reports}, 19\penalty0 (12):\penalty0 e46255, 2018.

\bibitem[Park et~al.(2021)Park, Choi, and Park]{park2021negative}
Beomjin Park, Hosik Choi, and Changyi Park.
\newblock Negative binomial graphical model with excess zeros.
\newblock \emph{Statistical Analysis and Data Mining: The ASA Data Science
  Journal}, 14\penalty0 (5):\penalty0 449--465, 2021.

\bibitem[Pratapa et~al.(2020)Pratapa, Jalihal, Law, Bharadwaj, and
  Murali]{pratapa2020benchmarking}
Aditya Pratapa, Amogh~P Jalihal, Jeffrey~N Law, Aditya Bharadwaj, and
  TM~Murali.
\newblock {Benchmarking algorithms for gene regulatory network inference from
  single-cell transcriptomic data}.
\newblock \emph{Nature Methods}, 17\penalty0 (2):\penalty0 147--154, 2020.

\bibitem[Ravikumar et~al.(2011)Ravikumar, Wainwright, Raskutti, and
  Yu]{ravikumar2011high}
Pradeep Ravikumar, Martin~J Wainwright, Garvesh Raskutti, and Bin Yu.
\newblock {High-dimensional covariance estimation by minimizing
  $\ell_1$-penalized log-determinant divergence}.
\newblock \emph{Electronic Journal of Statistics}, 5:\penalty0 935--980, 2011.

\bibitem[Shao(2003)]{shao2003mathematical}
Jun Shao.
\newblock \emph{{Mathematical statistics}}.
\newblock Springer Science \& Business Media, 2003.

\bibitem[Silva et~al.(2019)Silva, Rothstein, McNicholas, and
  Subedi]{silva2019multivariate}
Anjali Silva, Steven~J Rothstein, Paul~D McNicholas, and Sanjeena Subedi.
\newblock {A multivariate Poisson-log normal mixture model for clustering
  transcriptome sequencing data}.
\newblock \emph{BMC Bioinformatics}, 20\penalty0 (1):\penalty0 1--11, 2019.

\bibitem[Song et~al.(2022)Song, Shi, Meng, Ma, Huang, Zhang, Wu, Li, Lin, Yang,
  et~al.]{song2022single}
Guohe Song, Yang Shi, Lu~Meng, Jiaqiang Ma, Siyuan Huang, Juan Zhang, Yingcheng
  Wu, Jiaxin Li, Youpei Lin, Shuaixi Yang, et~al.
\newblock Single-cell transcriptomic analysis suggests two molecularly distinct
  subtypes of intrahepatic cholangiocarcinoma.
\newblock \emph{Nature Communications}, 13\penalty0 (1):\penalty0 1--15, 2022.

\bibitem[Specht and Li(2017)]{specht2017leap}
Alicia~T Specht and Jun Li.
\newblock {LEAP: constructing gene co-expression networks for single-cell
  RNA-sequencing data using pseudotime ordering}.
\newblock \emph{Bioinformatics}, 33\penalty0 (5):\penalty0 764--766, 2017.

\bibitem[Stuart et~al.(2019)Stuart, Butler, Hoffman,
  et~al.]{stuart2019comprehensive}
Tim Stuart, Andrew Butler, Paul Hoffman, et~al.
\newblock {Comprehensive integration of single-cell data}.
\newblock \emph{Cell}, 177\penalty0 (7):\penalty0 1888--1902, 2019.

\bibitem[Szklarczyk et~al.(2019)Szklarczyk, Gable, Lyon, Junge, Wyder,
  Huerta-Cepas, Simonovic, Doncheva, Morris, Bork,
  et~al.]{szklarczyk2019string}
Damian Szklarczyk, Annika~L Gable, David Lyon, Alexander Junge, Stefan Wyder,
  Jaime Huerta-Cepas, Milan Simonovic, Nadezhda~T Doncheva, John~H Morris, Peer
  Bork, et~al.
\newblock {STRING v11: protein--protein association networks with increased
  coverage, supporting functional discovery in genome-wide experimental
  datasets}.
\newblock \emph{Nucleic Acids Research}, 47\penalty0 (D1):\penalty0 D607--D613,
  2019.

\bibitem[Van~der Vaart(2000)]{van2000asymptotic}
Aad~W Van~der Vaart.
\newblock \emph{{Asymptotic statistics}}, volume~3.
\newblock Cambridge university press, 2000.

\bibitem[Wainwright et~al.(2008)Wainwright, Jordan,
  et~al.]{wainwright2008graphical}
Martin~J Wainwright, Michael~I Jordan, et~al.
\newblock {Graphical models, exponential families, and variational inference}.
\newblock \emph{Foundations and Trends in Machine Learning}, 1\penalty0
  (1--2):\penalty0 1--305, 2008.

\bibitem[Wille et~al.(2004)Wille, Zimmermann, Vranov{\'a}, F{\"u}rholz, Laule,
  Bleuler, Hennig, Preli{\'c}, von Rohr, Thiele, et~al.]{wille2004sparse}
Anja Wille, Philip Zimmermann, Eva Vranov{\'a}, Andreas F{\"u}rholz, Oliver
  Laule, Stefan Bleuler, Lars Hennig, Amela Preli{\'c}, Peter von Rohr, Lothar
  Thiele, et~al.
\newblock {Sparse graphical Gaussian modeling of the isoprenoid gene network in
  Arabidopsis thaliana}.
\newblock \emph{Genome biology}, 5\penalty0 (11):\penalty0 1--13, 2004.

\bibitem[Wu et~al.(2018)Wu, Deng, and Ramakrishnan]{Wu2018}
Hao Wu, Xinwei Deng, and Naren Ramakrishnan.
\newblock {Sparse estimation of multivariate Poisson log-normal models from
  count data}.
\newblock \emph{{Statistical Analysis and Data Mining: The ASA Data Science
  Journal}}, 11\penalty0 (2):\penalty0 66--77, 2018.

\bibitem[Xu et~al.(2013)Xu, Baroukh, Dannenfelser, Chen, Tan, Kou, Kim,
  Lemischka, and Ma'ayan]{xu2013escape}
Huilei Xu, Caroline Baroukh, Ruth Dannenfelser, Edward~Y Chen, Christopher~M
  Tan, Yan Kou, Yujin~E Kim, Ihor~R Lemischka, and Avi Ma'ayan.
\newblock {ESCAPE: database for integrating high-content published data
  collected from human and mouse embryonic stem cells}.
\newblock \emph{Database}, 2013, 2013.

\bibitem[Yakowitz and Spragins(1968)]{yakowitz1968identifiability}
Sidney~J Yakowitz and John~D Spragins.
\newblock {On the identifiability of finite mixtures}.
\newblock \emph{The Annals of Mathematical Statistics}, 39\penalty0
  (1):\penalty0 209--214, 1968.

\bibitem[Yang et~al.(2012)Yang, Ravikumar, Allen, and Liu]{yang2012graphical}
Eunho Yang, Pradeep Ravikumar, Genevera~I Allen, and Zhandong Liu.
\newblock {Graphical models via generalized linear models.}
\newblock In \emph{Advances in Neural Information Processing Systems},
  volume~25, pages 1367--1375, 2012.

\bibitem[Zhang et~al.(2020)Zhang, Liu, Zhang, Xie, Miao, Xia, and
  Guo]{zhang2020htftarget}
Qiong Zhang, Wei Liu, Hong-Mei Zhang, Gui-Yan Xie, Ya-Ru Miao, Mengxuan Xia,
  and An-Yuan Guo.
\newblock {hTFtarget: a comprehensive database for regulations of human
  transcription factors and their targets}.
\newblock \emph{Genomics, Proteomics \& Bioinformatics}, 18\penalty0
  (2):\penalty0 120--128, 2020.

\bibitem[Zhao and Yu(2006)]{zhao2006model}
Peng Zhao and Bin Yu.
\newblock {On model selection consistency of Lasso}.
\newblock \emph{The Journal of Machine Learning Research}, 7:\penalty0
  2541--2563, 2006.

\bibitem[Zheng et~al.(2017)Zheng, Terry, Belgrader, et~al.]{zheng2017massively}
Grace~XY Zheng, Jessica~M Terry, Phillip Belgrader, et~al.
\newblock {Massively parallel digital transcriptional profiling of single
  cells}.
\newblock \emph{Nature Communications}, 8\penalty0 (1):\penalty0 1--12, DOI:
  https://doi.org/10.1038/ncomms14049, 2017.

\bibitem[Zhou et~al.(2016)Zhou, Liu, Sun, Zheng, Zhou, Yang, and
  Qu]{zhou2016chipbase}
Ke-Ren Zhou, Shun Liu, Wen-Ju Sun, Ling-Ling Zheng, Hui Zhou, Jian-Hua Yang,
  and Liang-Hu Qu.
\newblock {ChIPBase v2. 0: decoding transcriptional regulatory networks of
  non-coding RNAs and protein-coding genes from ChIP-seq data}.
\newblock \emph{Nucleic Acids Research}, page gkw965, 2016.

\bibitem[Ziegenhain et~al.(2017)Ziegenhain, Vieth, Parekh,
  et~al.]{ziegenhain2017comparative}
Christoph Ziegenhain, Beate Vieth, Swati Parekh, et~al.
\newblock {Comparative analysis of single-cell RNA sequencing methods}.
\newblock \emph{Molecular Cell}, 65\penalty0 (4):\penalty0 631--643, 2017.

\end{thebibliography}

\end{document}